\RequirePackage{ifthen}

\newcommand{\typof}{1} %

\newcommand{\longv}[1]{\ifthenelse{\equal{\typof}{0}}{}{#1}}
\newcommand{\shortv}[1]{\ifthenelse{\equal{\typof}{0}}{#1}{}}

\newcommand{\longshortv}[2]{\ifthenelse{\equal{\typof}{0}}{#2}{#1}}
\newcommand{\drop}[0]{\ifthenelse{\equal{\typof}{0}}{}{}}

\shortv{
\documentclass[conference]{IEEEtran}}
\longv{
\documentclass[a4paper, 10pt]{article}
}

\shortv{

\ifCLASSINFOpdf
\else
\fi
}

\usepackage[utf8]{inputenc} 
\usepackage{stmaryrd}

\usepackage{amsmath}
\usepackage{amsthm}
\usepackage{amssymb}

\usepackage{hyperref}

\usepackage{caption}
\usepackage{subcaption}

\usepackage{xcolor}

\usepackage{cmll}

\usepackage{tikz}
\usetikzlibrary{cd}

\usepackage{bussproofs}
\EnableBpAbbreviations

\usepackage{graphicx}
\usepackage{adjustbox}

\usepackage{listings}

\longv{
\setlength\voffset{-1in}
\setlength\topmargin{0.5cm}
\setlength\headheight{1cm}
\setlength\headsep{0.5cm}
\setlength\textheight{24.7cm}
\setlength\footskip{0.5cm}
\setlength\hoffset{-1in}
\setlength\oddsidemargin{3cm}
\setlength\textwidth{15cm}
}

\newcommand{\TT}[1]{\mathtt{#1}}

\newcommand{\C}[1]{\mathcal{#1}}
\newcommand{\BB}[1]{\mathbb{#1}}

\newcommand{\To}{\Rightarrow}

\newcommand{\Real}{\mathsf{Real}}

\newcommand{\ff}[1]{#1_{\mathsf{fin}}}

\newcommand{\Lip}{\mathsf{Lip}}

\newcommand{\DDer}{\mathsf{D}}

\newcommand{\Met}{\mathsf{Met}}
\newcommand{\QQ}{\mathbf{Q}}
\newcommand{\QQR}{\mathbf{Q}^{\mathsf{r}}}
\newcommand{\QQS}{\mathbf{Q}^{\mathsf{s}}}
\newcommand{\QQRS}{\mathbf{Q}^{\mathsf{rs}}}
\newcommand{\QQU}{\mathbf{Q}_{\land}^{\mathsf s}}
\newcommand{\QQRU}{\mathbf{Q}^{\mathsf{rs}}_{\land}}

\newcommand{\Res}{\multimapinv}
\newcommand{\ResH}{\Leftarrow}

\newcommand{\COV}[3]{#1 \stackrel{#2}{\longrightarrow} #3}

\newcommand{\STLC}{\mathsf{ST\lambda C}}

\newcommand{\VPM}{\mathbf{V}}

\newcommand{\ELLE}{\mathbf{L}}
\newcommand{\LLE}{\mathbf{LL}}

\newcommand{\HLLip}{\mathbf{LL}^{*}}

\newcommand{\model}[1]{\llbracket #1\rrbracket}

\newcommand{\nudel}[1]{\llparenthesis #1\rrparenthesis}

\newcommand{\fini}[1]{\llparenthesis #1\rrparenthesis_{\mathrm{fin}}}

\newcommand{\ev}{\mathsf{ev}}

\definecolor{color0}{HTML}{4682B4}

\newtheorem{example}{Example}[section]
\newtheorem{definition}{Definition}[section]

\newtheorem{remark}{Remark}[section]
\newtheorem{theorem}{Theorem}[section]

\newtheorem{lemma}[theorem]{Lemma}
\newtheorem{proposition}[theorem]{Proposition}
\newtheorem{corollary}{Corollary}[section]

\longv{\title{
On Generalized Metric Spaces for \\ the Simply Typed Lambda-Calculus\thanks{
This work has been funded by the ERC CoG 818616 “DIAPASoN”.
}
\\ (Extended Version)}

\author{Paolo Pistone\\
Universit\`a di Bologna\\
\texttt{paolo.pistone2@unibo.it}
}

\date{}

}

\begin{document}
\longv{\maketitle}

\shortv{\title{
On Generalized Metric Spaces for \\ the Simply Typed Lambda-Calculus
\thanks{
This work has been funded by the ERC CoG 818616 “DIAPASoN”.
}
}
}


\shortv{\author{
\IEEEauthorblockN{Paolo Pistone}
\IEEEauthorblockA{Universit\`a di Bologna\\
Email:paolo.pistone2@unibo.it}}

\IEEEoverridecommandlockouts
\IEEEpubid{\makebox[\columnwidth]{978-1-6654-4895-6/21/\$31.00~
\copyright2021 IEEE \hfill} \hspace{\columnsep}\makebox[\columnwidth]{ }}

  \IEEEcompsocthanksitem
\maketitle
}

\begin{abstract}

Generalized metrics, arising from Lawvere's view of metric spaces as enriched categories, have been widely applied in denotational semantics as a way to measure \emph{to which extent} two programs behave in a similar, although non equivalent, way. 
However,  the application of generalized metrics to higher-order languages like the simply typed lambda calculus has so far proved unsatisfactory.
In this paper we investigate a new approach to the construction of cartesian closed categories of generalized metric spaces.
Our starting point is a quantitative semantics based on a generalization of usual logical relations. Within this setting, we show that several families of generalized metrics provide ways to
extend the Euclidean metric to all higher-order types.

\end{abstract}


\shortv{
\IEEEpeerreviewmaketitle
}

\section{Introduction}
In the literature on program semantics
much attention has been devoted to program equivalence, and, accordingly, to the study of program transformations which do not produce observable changes of behavior. 
However, in fields involving {numerical} or {probabilistic} forms of computation one often deals with transformations that \textit{do} alter program behavior, replacing a piece of program with one which is only approximately equivalent. 
For example, numerical methods (e.g.~linear regression, numerical integration) are based on the replacement of computationally expensive operations with more efficient, although less precise, ones. 
On another scale, statistical learning algorithms compute approximations of a desired function by fitting with a finite sample. 

The challenge that accompanies the use of such \emph{approximate} program transformations \cite{chaudhuri} is to come up with methods to measure and bound the error they produce.
This has motivated much literature on \emph{program metrics} \cite{ARNOLD1980181, VANBREUGEL20011,Gabo2019, Escardo1999, BAIER1994171,DalLago2015, 10.1007/978-3-662-44584-6_4, 10.1007/978-3-662-54434-1_13, 10.1145/3209108.3209149}, that is, on semantics in which types are endowed with a notion of distance. 
This approach has found widespread applications, for example in differential privacy \cite{Gaboardi2017, 10.1007/978-3-642-29420-4_3, Barthe_2012} and reinforcement learning \cite{Panangaden2011}.

A natural framework for the study of program metrics and their abstract properties is provided by so-called \emph{generalized metrics}. Since Lawvere's \cite{Lawvere1973} it has been known that some of the basic axioms of standard metric spaces (notably, the \emph{reflexivity} and \emph{transitivity} axioms $d(x,x)=0$ and $d(x,z)+d(z,y)\geq d(x,y)$) can be seen, at a higher level of abstraction, as describing the structure of a category \emph{enriched} over some quantitative algebra. Typically, when this algebra is the usual semi-ring of positive reals (i.e.~when ``0'' actually means zero, and ``+'' actually means plus), one gets the metric spaces everyone is used to. However, one can consider generalized distance functions $d: X\times X\to Q$, where $Q$
is now a different algebra (typically a \emph{quantale} or a \emph{quantaloid} \cite{Hofmann2014}), 
and the monoidal structure of $Q$ determines the actual meaning of the metric axioms. 
Well-investigated examples of this generalized approach are given by \emph{ultra}-metric spaces \cite{VANBREUGEL20011,Escardo1999}, \emph{partial} metric spaces \cite{matthews, Bukatin1997,AGT7849,Stubbe2018} and \emph{probabilistic} metric spaces \cite{Sklar1983, Hofmann:2013aa}.

Generalized program metrics have been applied in several areas of computer science, e.g.~to co-algebraic \cite{DBLP:journals/lmcs/BaldanBKK18,Pouzet2020}, and concurrent \cite{10.1007/978-3-662-44584-6_4} systems, and to algebraic effects \cite{Plotk,10.1145/3209108.3209149}. 
%
However, the application of program metrics to even basic \emph{higher-order} languages like the simply typed $\lambda$-calculus $\STLC$ has so far proved unsatisfactory. 
One can mention both theoretical and practical reasons for this failure. At the abstract level, for instance, there is the well-known fact that standard categories of metric spaces, even generalized, are usually \emph{not} cartesian closed, and thus only account for linear or sub-exponential variants of $\STLC$ \cite{10.1145/1932681.1863568, Gaboardi_2013, Gaboardi2017}.
At a more practical level, there is the observation that 
%
 even with such restrictions, the distance between two functional programs computed in such models is often  not very informative, as it estimates the error of replacing one program by the other one \emph{in the worst case}, and thus independently of the current \emph{context} in which these programs are placed. 
 

%
%

In this paper we introduce a new class of program metric semantics for $\STLC$ which overcome the aforementioned difficulties. 
These semantics arise from the study of a class of quantitative models based on what we call \emph{quantitative logical relations} (in short, QLR). 
A QLR is just what remains of a generalized metric space when one discards the reflexivity and transitivity axioms; in other words, it is nothing more than a function $a: X\times X\to Q$ relating pairs of points $x,y\in X$ with an element $a(x,y)$ of some quantitative algebra $Q$. 
At the same time, such functions can be seen as quantitative analogs of standard logical relations. The difference is that while with the latter
two programs may or may not be related, with QLR two programs are always related \emph{to a certain degree}. 
 
We believe that models for $\STLC$ should be as elementary as possible. By the way, the category of sets is itself a denotational model of $\STLC$. For this reason, we do not, at first, impose any restriction (e.g.~continuity, Lipschitz continuity) over the set-theoretic functions between QLR. 
Importantly, maps of QLR can relate functions measuring distances over \emph{different} quantitative algebras. For this reason, set-theoretic maps are accompanied by a second map, a sort of \emph{derivative}, relating \emph{errors} in input with errors in output.
 This idea, which extends similar ones from \emph{differential logical relations} \cite{dallago,dallago2} and \emph{diameter spaces} \cite{Geoffroy2020}, 
  mark the main difference between our approach and usual metric semantcs (in which one usually considers a \emph{fixed} quantale), and is a key ingredient to obtain models of the full $\STLC$.

Our first contribution is to show that several variants of QLR form cartesian closed categories and that some standard results about logical relations have a quantitative analog in the realm of QLR. These results show that QLR-models capture quantitative relational reasoning of higher-order programs in a fully compositional way.

However, recall that our starting point was program metric semantics, and QLR, by their very definition, are \emph{not} metric spaces. 
Yet, 
since generalized metrics are particular cases of QLR, the latter provide an ideal environment 
to investigate \emph{which} families of generalized metrics
(i.e.~which choices of  the ``0'' and the ``+'')
 adapt well to the cartesian closed structure.

Our second contribution is a characterization of the class of generalized metric spaces that give rise to cartesian closed categories of QLR. These results demonstrate the existence of a variety of compositional metric semantics of $\STLC$ which extend the Euclidean metrics over the reals to all simple types.

Finally, we show that the derivatives found in QLR-models can be compared with those appearing in other quantitative models of $\STLC$, like those arising from the \emph{differential $\lambda$-calculus} \cite{ER,Blute2009, BUCCIARELLI2010213}.

\medskip
\paragraph*{Outline} 
After motivating the introduction of QLR in Section \ref{section2}, in Section \ref{section3} we recall the definition of some classes of generalized metric spaces; in Section \ref{section4} we introduce two cartesian closed categories $\QQ$ and $\QQR$ of QLR, and we describe the interpretation of $\STLC$ in them. In Section \ref{section5} we investigate the generalized metrics which form cartesian closed sub-categories of $\QQ$ and $\QQR$.
Finally, in Section \ref{section5} we construct a different cartesian closed category $\LLE_{\mathsf{Met}}$ of generalized metric spaces based on a ``locally Lipschitz'' condition for QLR morphisms.

\section{Higher-Order Metric Semantics}\label{section2}


\subsection{Program Metrics and Higher-Order Languages}

Program metrics have been widely investigated to capture properties like program \emph{similarity} and \emph{sensitivity}. 
The fundamental idea is usually to associate types $\sigma,\tau$ with metric spaces, and programs $f:\sigma\to \tau$ with \emph{non-expansive}, or more generally \emph{Lipschitz continuous} functions. This means that for all programs $t,u$ of type $\sigma$, the distance between $f(t)$ and $f(u)$ does not exceed that between $t$ and $u$ by more than a fixed factor $L$ (formally, $d(f(t),f(u))\leq L\cdot d(t,u)$).

%

However, the approach just sketched is not satisfactory for the interpretation of \emph{higher-order} languages, as those based on $\STLC$. The main problem is that the category $\Met_{Q}$ of metric spaces over a quantale $Q$ and non-expansive maps \cite{Hofmann2014}, which provides the abstract setting for usual program metrics, is not compatible with the usual structure of models of $\STLC$.
More precisely, while the space $ \Met_{Q}(X,Y)$ of non-expansive functions can be endowed with a metric (the $\sup$-metric $d_{\sup}(f,g)=\sup\{d(f(x),g(x))\mid x\in X\}$), this construction does \emph{not} yield a right-adjoint to the categorical product. For this reason $\Met_{Q}$ is not a \emph{cartesian closed} category (although $\Met_{Q}$ still admits some interesting cartesian closed \emph{sub}-categories, see \cite{CLEMENTINO20063113, Clementino:2009aa}). 
%

This abstract issue is not the only one has to face, though.
After all, category theory is usually invoked in program semantics as a way to enforce \emph{compositionality}, i.e.~the property by which the semantics of a composed program is expressed in terms of the semantics of its components.
Yet, even if we accept to restrict ourselves to higher-order languages compatible with the categorical structure of $\Met_{Q}$ (like e.g.~the system $\mathsf{Fuzz}$ \cite{10.1145/1932681.1863568}), the metric $d_{\sup}$ still does not account for the behavior of higher-order programs in a sufficiently {compositional}, and, in the end, informative way. For example, as observed in \cite{dallago}, consider the two Lipschitz functions $f=\lambda x.\sin(x)$ and $g=\lambda x.x$: since $f$ and $g$ get arbitrarily far from each other in the worst case (i.e.~as $x$ approaches $\pm\infty$), one can deduce that $d_{\sup}(f,g)$ is infinite. Hence, the distance $d_{\sup}(f,g)$ provides no significant information in any situation in which $f$ is replaced by $g$ as a {component} of a larger program: for instance, if $\TT C[\ ]$ is a context applying a function on values close to 0, the programs $\TT C[f]$ by $\TT C[g]$ will likely turn out close, yet there is no way to predict this fact on the basis of $d_{\sup}(f,g)$.

A related issue occurs with \emph{contextual} notions of distance, as those found e.g.~in probabilistc extensions of the $\lambda$-calculus \cite{DalLago2015}. These metrics extend usual contextual equivalence, by letting the distance $d_{\mathsf{ctx}}(t,u)$ between two objects of type $\sigma$ be the sup of all observable distances $d_{\mathsf{Euc}}(\TT C[f], \TT C[g])$, for any context $\TT C[\ ]:\sigma\To \Real$. As shown in \cite{10.1007/978-3-662-54434-1_13}, the non-linearity of $\STLC$ can be used to define contexts that arbitrarily {amplify} distances, with the consequence that the metric $d_{\mathsf{ctx}}$ trivializes onto plain contextual equivalence.

\subsection{From Program Metrics to Quantitative Logical Relations}

To overcome these issues, in Section \ref{section4} we introduce \emph{quantitative logical relations}, a quantitative extension of usual logical relations (generalizing previous approaches \cite{dallago,dallago2, Geoffroy2020}) which, on the one hand, applies to higher-order programs without restrictions (e.g.~Lipschitz-continuity), and, on the other hand, enables reasoning about behavioral similarity in a 
fully compositional way.

%


Semantically, logical relations for a programming language $\C L$ can be introduced starting from a denotational model of $\C L$ (for simplicity, we consider a simple set-theoretic model, associating each type $\sigma$ with a set $\model \sigma$ and each program $t:\sigma\to \tau$ with a function $\model t:\model \sigma \to \model \tau$); one then constructs a more refined model whose objects are binary relations $r: \model \sigma \times \model \sigma\to\{0,1\}$, and whose arrows are those functions from our original model which send related points into related points (in more abstract terms, this construction is an instance of the \emph{glueing} construction, see \cite{Schalk2003}). The so-called \emph{Fundamental Lemma} tells us then that any program $t:\sigma\to \tau$ of $\C L$ yields a morphism in this model, i.e.~preserves relatedness.

While in logical relations relatedness is measured over a fixed algebra (the Boolean algebra $\{0,1\}$), in QLR relatedness is measured over a larger class of quantales.  Hence, a QLR is of the form 
$a: \model \sigma \times \model \sigma\to \nudel \sigma$, where $\nudel \sigma$ is some quantale associated with $\sigma$. Typically, when $\sigma$ is a functional type, $\nudel{\sigma}$ will be some quantale of functions mapping differences in input into differences in output. 

%
%

%

To interpret a program $t:\sigma\to \tau$ we must accompany the function $\model t$ with a second function $\nudel t: \model \sigma \times \nudel \sigma  \to \nudel \tau$ mapping differences in $\nudel\sigma$ \emph{around some point of $\model\sigma$} into differences in $\nudel \tau$. 
The function $\nudel t$ can be seen as sort of \emph{derivative} of $\model t$, and is the key ingredient to reason about $t$ in a compositional way: if $\alpha\in \nudel \sigma$ measures the similarity of two programs $u,v$ and $\TT C[\ ]:\sigma \to \tau$ is a context with derivative $\nudel{\TT C}$, then by composing $\nudel{\TT C}$ with $\model u$ and $\alpha$, we obtain a measure of the similarity between $\TT C[u]$ and $\TT C[v]$. 
Notably, the Fundamental Lemma of logical relations translates in this setting into a result showing that any program $t$ from $\C L$ translates into a derivative $\nudel t$, yielding a fully compositional semantics for $\C L$.

For instance, take $\model \Real=\BB R$ and $\nudel \Real=\BB R_{\geq 0}$; if $f,g: \Real\to \Real$ are the two programs $\lambda x.\sin(x), \lambda x.x$ seen before and $\TT C[\ ]=[\ ]0:(\Real\to \Real)\to \Real$ is the context that applies a function to $0$, in our setting we can reason as follows:
first, the difference $d(\model f, \model g)$ will be itself a function mapping small differences in input \emph{around} 0 onto small differences in output; secondly, the derivative $\nudel{\TT C}$ 
will be such that that the value $\nudel{\TT C}(f,\varphi)$ only depends on how much $\varphi$ grows on small neighborhoods of $0$; hence, the difference between $\TT C[f]$ and $\TT C[g]$, computed by applying
$\nudel{\TT C}$ to $\model f $ and to $d(\model f, \model g)$, will yield a value close to 0.

Similar ideas already appear in \cite{dallago2, Geoffroy2020} and have been shown to 
provide a compositional account of techniques from \emph{incremental computing} and
\emph{approximate programming} (e.g.~{loop perforation} \cite{loopperf} and {numerical integration}).
The study of QLR, that we develop here, is intended to capture 
the basic structure underlying such (non-equivalent) constructions, and to characterize a much larger family of quantitative and metric models to which those from \cite{dallago2, Geoffroy2020} belong.

\subsection{...and back to Generalized Metric Spaces}

While a QLR $a: \model \sigma \times \model \sigma \to \nudel\sigma$ needs not be a metric, several classes of generalized metric spaces can be seen as QLR satisfying further properties.
One can thus ask \emph{which} families of generalized metrics can be lifted to all simple types within a given QLR-model.

In Section \ref{section5} we investigate generalized metrics in categories of QLR with unrestricted morphisms (that is, with no continuity or Lipschitz restriction). We show that, under some mild assumptions, 
lifting metrics to simple types forces distances to be idempotent (i.e.~to satisfy $\alpha=\alpha+\alpha$). This implies that the generalized metrics that can be lifted to all simple types are of two kinds: firstly,  the \emph{ultra-metric} and \emph{partial ultra-metric} spaces, that is, those metrics based on an idempotent quantitative algebra; secondly, 
those generalized metrics whose distance function can be \emph{factorized} through an idempotent metric. By extending a construction from \cite{Geoffroy2020} relating partial metrics with lattice-valued distances, we show that the Euclidean metric, as well as many other standard {metrics} and {partial metrics}, belong to this second class.

In Section \ref{section6} we investigate generalized metrics in categories of QLR where morphisms satisfy suitable generalizations of the \emph{Lipschitz} and \emph{locally Lipschitz} continuity conditions. 
We first show that the first condition does \emph{not} yield a cartesian closed category, 
for reasons very similar to those found when considering metrics over a fixed quantale.
We then show that the second does yields, instead, a model of $\STLC$ in which types are interpreted by generalized metric spaces.

\section{Generalized Metric Spaces}\label{section3}

In this paper we consider several variants of metric spaces. It is thus useful to adopt a general and abstract definition of what we take a (generalized) metric space to be. 
We exploit the abstract formulation of generalized metric spaces as enriched categories dating back to Lawvere's \cite{Lawvere1973}, 
who first observed that a metric space in the standard sense can be seen as a category enriched in the monoidal poset $([0,+\infty), \geq, 0,+)$ of positve real numbers under reversed ordering and addition.

%
%

\subsection{Metrics over an Arbitrary Quantale}

The standard axioms of metric spaces involve an order relation and a monoidal operation (addition) with a neutral element 0. This structure is characterized by a \emph{monoidal poset}, that is, a tuple $(M,\geq, 0,+)$ where $(M,\geq)$ is a poset and $(M,0,+)$ is a monoid such that $+$ is monotone.  
 In practice, one is usually interested in measuring distances in monoidal posets where $\sup$s and $\inf$s always exist. This leads to consider (commutative and integral) quantales:
 
 \begin{definition}
 A (commutative) quantale is a commutative monoidal poset $(Q,0,+,\geq)$ such that $(Q,\geq)$ is a complete lattice satisfying
$ \alpha +\bigwedge S= \bigwedge\{\alpha+\beta\mid \beta\in S\}$, for all $S\subseteq Q$.
%
A quantale $(Q,0,+,\geq)$ is \emph{integral} when $0=\bot$.
 A commutative quantale $Q$ is a \emph{locale} when $0=\bot$ and $\alpha=\alpha+\alpha$ holds for all $\alpha\in Q$ (or, equivalently, when $\alpha+\beta=\alpha\lor\beta$).

 \end{definition}

 \begin{remark}
With respect to common presentations of quantales, we adopt here the \emph{reversed} order (so that $\bigvee$s and $\bigwedge$s are inverted), as this is more in accordance with the quantitative intuition.
 \end{remark}

 \begin{example}[The Lawvere quantale]
The structure $(\BB R^{\infty}_{\geq 0}, 0, +, \leq)$, where $\BB R^{\infty}_{\geq 0}$ is the set of positive reals plus $\infty$, is a commutative and integral quantale, and is usually referred to as the \emph{Lawvere quantale} \cite{Hofmann2014}.
 If we replace $+$ with $\sup$, the resulting structure  $(\BB R^{\infty}_{\geq 0}, 0, \sup, \leq)$ is a locale.
\end{example}

\begin{example}\label{example:subsets}
For any commutative monoid $(M,0,+)$, the structure $(\wp(M), \{0\},+,\subseteq)$, is a commutative quantale, where $A+B=\{x+y\mid x\in A, y\in B\}$. 
\end{example}

\begin{example}
All products $\Pi_{i\in I}Q_{i}$ of (commutative and integral) quantales, with the pointwise order, are still commutative and integral quantales.
\end{example}

In a quantale $Q$ one can define the following two operations:
\begin{align*}
\alpha \Res \beta  = \bigwedge \{\delta\mid \beta + \delta \geq \alpha\}\qquad
\alpha \ResH \beta  = \bigwedge \{\delta\mid \beta \vee \delta \geq \alpha\}
\end{align*}
In any quantale $\delta \geq \alpha \Res \beta$ holds iff $\delta + \beta \geq \alpha$, that is, $\Res$ is \emph{right-adjoint} to $+$. A quantale in which $\ResH$ is right-adjoint to $\vee$, i.e.~$\delta \geq \alpha \ResH \beta$ holds iff $\delta \vee \beta \geq \alpha$, is called a \emph{Heyting quantale} \cite{Hofmann2014, CLEMENTINO20063113}.
The Lawvere quantale and all other quantales obtained from it by product are Heyting. Moreover, all locales are Heyting.

\begin{example}
In the Lawvere quantale  $x\Res y= \max\{0, x-y\}$ and $x\ResH y$ is $0$ if $x\leq y$ and is $x$ otherwise.
\end{example}

Over any quantale $Q$ we can define generalized metric spaces as follows:

\begin{definition}\label{def:gms1}
A \emph{generalized metric space} is a triple $(X,Q,a)$ where $X$ is a set, 
 $Q$ is a commutative quantale,
and $a:X\times X\to Q$ satisfies, for all $x,y,z\in X$:
\begin{align} 
0 & \geq a(x,x)  \tag{reflexivity}\\
a(x,y)+a(y,z) & \geq a(x,z) \tag{transitivity}
\end{align}

A generalized metric space is said:
\begin{itemize}
\item \emph{symmetric} if $a(x,y)=a(y,x)$;
\item \emph{separated} if $a(x,y)=0$ implies $x=y$.
\end{itemize}
\end{definition}


Observe that, when $Q$ is integral, from the reflexivity axiom it follows that $a(x,x)=0$ holds for all $x\in X$.


Following usual terminology, we let a \emph{pseudo-metric space} be a symmetric metric space $(X,Q,a)$, and 
a \emph{standard metric space} be a separated pseudo-metric space.


The \emph{Euclidean metric} is the standard metric space $(\BB R, \BB R^{\infty}_{\geq 0}, d_{\mathsf{Euc}})$ where $d_{\mathsf{Euc}}(x,y)=|x-y|$.

\begin{example}
A standard metric space $(X,Q,a)$ in which $Q$ is a locale is usually called a \emph{ultra-metric space}. The transitivity axiom reads in this case as $a(x,y)\vee a(y,z)\geq a(x,z)$.
For instance, the \emph{sequence metric} on the set  $X^{\BB N}$ of $X$-sequences $(x_{n})_{n\in \BB N}$ is the ultra-metric space $(X^{\BB N}, \BB R^{\infty}_{\geq 0}, d_{\mathsf{seq}})$ given by $d_{\mathsf{seq}}(x_{n},y_{n})=2^{-c(x_{n},y_{n})}$, where $c(x_{n},y_{n})$ is the length or the largest common prefix of $x_{n}$ and $y_{n}$. 

\end{example}

\begin{example}\label{ex:prob}
A standard metric space $(X,\Delta,a)$ in which $\Delta$ is the quantale of \emph{distributions}, i.e.~the left-continuous maps $f: \BB R_{\geq 0}\to[0,1]$ with pointwise ordering and monoidal operation $(f\oplus g)(r)=\bigwedge_{s+t=r}f(s)\cdot g(t)$, is an example of \emph{probabilistic metric space} \cite{Sklar1983, Hofmann:2013aa}. Observe that the transitivity axiom reads in this case as $a(x,y)(r)+a(y,z)(s)\geq a(x,y)(r+s)$.

\end{example}

\subsection{Partial Metric Spaces}

In several approaches to program metrics one encounters distance functions which do not satisfy the reflexivity axiom $0\geq a(x,x)$.
A basic example (see \cite{matthews}) is obtained when the sequence metric
$d_{\mathsf{seq}}$ is extended to the set $\widehat X=\bigcup_{n}^{\infty}X^{n}\cup X^{\BB N}$ of \emph{finite and infinite} $X$-sequences (this kind of spaces are common, for instance, in domain theory): whenever $x_{n}$ is a sequence of length $k$, we have that $d_{\mathsf{seq}}(x_{n},x_{n})=2^{-k}>0$. 

The simplest way to define a metric with non-zero self-distances is simply to drop the reflexivity axiom. This yields the \emph{relaxed metrics} from \cite{Bukatin1997}. An even more drastic relaxation of the metric axioms is the one considered in \cite{dallago}, where transitivity is also weakened to\footnote{Actually, \cite{dallago} does not define a distance function $d:X\times X\to Q$ but rather a distance \emph{relation} $\rho \subseteq X\times Q\times X$ obeying a relaxed transitivity of the form $\rho(x,\alpha,y), \rho(y,\beta,y), \rho(y,\gamma,z) \Rightarrow \rho(x,\alpha+\beta+\gamma,y)$. In fact, this is the same thing as a  function $d_{\rho}: X\times X\to \wp(Q)$ (where $\wp(Q)$ indicates the quantale of subsets of $Q$ from Example~\ref{example:subsets}) satisfying \eqref{eq:superrelaxed}.\label{foot1}}
\begin{equation}\label{eq:superrelaxed}
a(x,y)\leq a(x,z)+a(z,z)+a(z,y)
\end{equation}
We will refer to the latter as \emph{hyper-relaxed metrics}.

A different approach consists in considering distance functions that \emph{do} satisfy both metric axioms, but relative to a \emph{different} monoidal structure over $Q$.
The \emph{partial metric spaces} \cite{matthews,Bukatin1997}, developed to account for domains of objects akin to the set $\widehat X$, provide an example of this approach, 
%
as shown by the elegant presentation from \cite{Stubbe2018, STUBBE201495}, that we recall below. 

For any commutative integral quantale $Q$, let $\C D(Q)$ be the category whose objects are all elements of $Q$, and where $\C D(Q)(\alpha,\beta)$ is the complete lattice of \emph{diagonals} from $\alpha$ to $\beta$, i.e.~those $\delta\in Q$ satisfying
$$\alpha+ (\delta\Res \alpha) =\delta= (\delta\Res \beta)+\beta$$
The identity morphism $\mathrm{id}_{\alpha}$ is just $\alpha$ (moreover, $\alpha$ is the smallest element of $\C D(Q)(\alpha,\alpha)$); the composition of two diagonals $\delta\in \C D(Q)(\beta,\alpha)$ and $\eta\in \C D(Q)(\gamma,\beta)$ is the diagonal 
$$
\eta+_{\beta} \gamma := \eta +(\gamma \Res \beta)\in \C D(Q)(\gamma,\alpha)
$$
The category $\C D(Q)$ is an example of \emph{quantaloid} (see~\cite{STUBBE201495}).

\begin{example}
In the Lawvere quantale, a diagonal from $x$ to $y$ is any real number $z\geq x,y$, and the composition law reads as $x+_{z}y:= x+y-z$.
\end{example}
\begin{remark}\label{rem:loca}
When $Q$ is a locale, $\C D(Q)(\alpha,\beta)=\{\gamma\mid \alpha\vee \beta \leq \gamma\}$ and 
the composition law of $\C D(Q)$ coincides with that of $Q$, since 
$\alpha\vee (\beta\ResH \gamma)= \alpha\vee \beta$ holds for all $\gamma \leq \beta$.

Using this fact, the definition of the category of diagonals can be extended to the case in which $Q$ is just a complete lattice (and thus needs not be a locale), by letting $\C D(Q)(\alpha,\beta)=\{\gamma\mid \alpha\vee \beta \leq \gamma\}$, with identities $\mathrm{id}_{\alpha}=\alpha$ and composition given by $\lor$.
 The category $\C D(Q)$ is then a quantaloid precisely when $Q$ is a locale.
\end{remark}

%

%
%
%
%


Partial metric spaces can be defined as metric spaces with respect to the monoidal structure of diagonals:
\begin{definition}
A \emph{partial metric space} is a tuple $(X,Q,t,a)$ where $X$ is a set, $Q$ is a (commutative and integral) quantale, $t:X\to Q$ and $a:X\times X\to Q$ are such, for all $x,y,z\in X$, $a(x,y)\in \Delta(Q)(ty,tx)$ and:
\begin{align}
\mathrm{id}_{tx} & \geq a(x,x) \tag{reflexivity} \\
a(x,y) +_{ty} a(y,z) & \geq a(x,z) \tag{transitivity}
\end{align}

A partial metric space is said:
\begin{itemize}
\item \emph{symmetric} if $a(x,y)=a(y,x)$;
\item \emph{separated} if $a(x,y)=a(x,x)=a(y,y)$ implies $x=y$.

\end{itemize}
\end{definition}

\begin{remark} 
When $Q$ is integral, reflexivity forces $tx=a(x,x)$, so the partial metric structure is entirely determined by the triple $(X,Q,a)$.
\end{remark}

%
%

A symmetric and separated partial metric over the Lawvere quantale $a:X\times X\to \BB R^{\infty}_{\geq 0}$ satisfies the axioms below:
\begin{description}
\item[PMS1] \ $a(x,x)\leq a(x,y),a(y,x)$;
\item[PMS2] \ $a(x,y)=a(y,x)$;
\item[PMS3] \ if $a(x,x)=a(x,y)=a(y,x)$, then $x=y$;
\item[PMS4] \ $a(x,y)\leq a(x,z) + a(z,y) - a(z,z) $.
\end{description}

Observe that a  (symmetric and separated) metric is the same as a  (symmetric and separated) partial metric with $a(x,x)=0$. Moreover, any (symmetric and separated) partial metric $a:X\times X\to Q$ gives rise to a (symmetric and separated)  metric 
\begin{align*}
a^{*}(x,y)=
(a(x,y)\Res a(x,x))+(a(x,y)\Res a(y,y))
\end{align*}
The terminology for \emph{pseudo-}, \emph{standard} and \emph{ultra-}metrics extends straightforwardly to from metric to partial metric spaces.
%

For example, the sequence metric $d_{\mathsf{seq}}$ extended to $\widehat X$ yields a partial ultra-metric space.
Another standard example of partial metric over the Lawvere quantale is the one defined over the set $\C I$ of \emph{closed intervals} $\{[r,s]\mid r\leq s\}$ by $p([r,s],[r',s'])= \max\{s,s'\}- \min\{r,r'\}$. 

\begin{remark}
The definition of partial ultra-metric spaces can be extended, as we will do in Section \ref{section5}, to the case in which $Q$ is just a complete lattice, using Remark \ref{rem:loca}. However, one must be careful that all properties that rely on the existence of the right-adjoint $\ResH$ need not hold in this case.
\end{remark}

\section{Quantitative Logical Relations}\label{section4}

In this section we introduce two categories $\QQ$ and $\QQR$ of quantitative logical relations. After describing their cartesian closed structure, we describe the interpretation of $\STLC$ in these categories and we 
show that some standard results about logical relations scale to QLR in a quantitative sense. 

\subsection{Two Categories of QLR}

A \emph{quantitative logical relation} $(X,Q,a)$ (in short, a \emph{QLR}) is the given of a set $X$, a commutative quantale $Q$ and a function $a: X\times X\to Q$. A \emph{map of quantitative logical relations} $(X,Q,a)$, $(Y,R,b)$ is a pair $(f,\varphi)$, where $f:X\to Y$, $\varphi: X\times Q\to R$ and for all $x,y\in X$,
$$
a(x,y)\leq \alpha  \ \To \ b(f(x),f(y)) \leq \varphi(x,\alpha)
$$
QLR and their maps form a category $\QQ$ having as identities the pairs $(\mathrm{id}_{X}, \lambda x\alpha.\alpha)$, and composition defined by $(g,\psi)\circ (f,\phi)= (g\circ f, \psi\circ \langle f\circ \pi_{1}, \varphi\rangle)$.

The category $\QQ$ is cartesian closed: given QLR $(X,Q,a)$ and $(Y,R,b)$,  their cartesian product is the QLR $(X\times Y, Q\times R, a\times b)$, with unit $(\{\star\}, \{\star\}, \langle\star,\star\rangle\mapsto \star)$, and 
 their exponential is the QLR $(Y^{X}, R^{X\times Q},d^{\QQ}_{a,b})$ where
\shortv{ \begin{center}
\resizebox{0.48\textwidth}{!}{
$
d^{\QQ}_{a,b}(f,g)(x,\alpha)  =  \sup  \{d(f(x),g(y)),  d(f(x),f(y)) \mid a(x,y)\leq \alpha\}
$
%
}
\end{center}}
\longv{
$$
d^{\QQ}_{a,b}(f,g)(x,\alpha)  =  \sup  \{d(f(x),g(y)),  d(f(x),f(y)) \mid a(x,y)\leq \alpha\}
$$
}
The isomorphism $\begin{tikzcd}\QQ(Z\times X, Y)\ar[bend left=5]{r}{\lambda} & \QQ(Z,Y^{X}) \ar[bend left=5]{l}[below]{\ev}\end{tikzcd}$ defining the cartesian closed structure is given by $\lambda(f,\varphi)=(\lambda (f),\lambda(\varphi))$ and $\ev(f,\varphi)=(\ev(f),\ev(\varphi))$, where
\begin{align*}
\lambda (f)(z)(x) & =f(\langle z,x\rangle)\\ 
\lambda (\varphi)(\langle z,\gamma\rangle)(\langle x, \alpha\rangle)&=
\varphi(\langle\langle z,x\rangle, \langle \gamma,\alpha\rangle\rangle)\\
\ev(f)(\langle z,x\rangle)&=f(z)(x)\\
\ev(\psi)(\langle\langle z,x\rangle, \langle \gamma,\alpha\rangle\rangle)& = 
\psi(\langle z,\gamma\rangle)(\langle x, \alpha\rangle) 
\end{align*}

%

Given QLR $(X,Q,a)$ and $(Y,R,b)$, for any function $f:X\to Y$ there exists a \emph{smallest} function $\DDer(f):X\times Q\to R$ such that $(f,\DDer(f))\in \QQ (X,Y)$, defined by
\begin{align}
\DDer(f)(x,\alpha)= \sup \{ b(f(x),f(y))\mid a(x,y)\leq \alpha\}
\end{align}
We call $\DDer(f)$ the \emph{derivative} of $f$. Derivatives in $\QQ$ satisfy the following properties:
\begin{align}
\DDer(\mathrm{id}_{X})(x, \alpha) & = \alpha \tag{D1} \label{eq:prop2}\\
\DDer( \pi_{i})(\langle x_{1},x_{2}\rangle, \langle \alpha_{1},\alpha_{2}\rangle) & =
\alpha_{i} \tag{D2}\label{eq:prop3}\\
\DDer(\langle f,g\rangle)(x,\alpha) & =\langle \DDer(f)(x,\alpha), \DDer(g)(x,\alpha)\rangle \tag{D3}\label{eq:prop4}\\
%
\DDer(g\circ f)(x,\alpha) & \leq \DDer(g)(f(x),\DDer(f)(x,\alpha)) \tag{D4}\label{eq:prop5} \\
\DDer(\lambda(f))(x, \alpha) & \leq \lambda (\DDer(f))(x,\alpha) \tag{D5}\label{eq:prop6} \\
\DDer(\ev(f))(x,\alpha) & \leq \ev(\DDer(f))(x,\alpha) \tag{D6}\label{eq:prop7}
\end{align}

Properties \eqref{eq:prop2}-\eqref{eq:prop4} recall some of the axioms of \emph{Cartesian Differential Categories} \cite{Blute2009}, a well-investigated formalization of abstract derivatives. Property \eqref{eq:prop5} is a \emph{lax} version of the chain rule, and properties \eqref{eq:prop6} and \eqref{eq:prop7} state that $\DDer$ commutes with the cartesian closed isomorphisms in a lax way.


\begin{remark}
Derivatives $\partial(f)$ in Cartesian Differential Categories are \emph{additive} in their second variable, i.e.~they satisfy $\partial(f)(x,0)=0$ and $\partial(f)(x,\alpha+\beta)=\partial(f)(x,\alpha)+\partial(x,\beta)$. By contrast, it is not difficult to construct counter-examples to the additivity of $\DDer(f)$.
Let $f,g:\BB R\to \BB R$ be given by
$$
f(x)=
\begin{cases}
x & \text{ if } |x|\leq 1 \\
2x & \text{ otherwise}
\end{cases}
\qquad
g(x)=
\begin{cases}
2x & \text{ if } |x|\leq 1 \\
x & \text{ otherwise}
\end{cases}
$$
Then $3=\DDer(f)(0,1+1)> \DDer(f)(0,1)+\DDer(f)(0,1)=2$ and 
$ 3= \DDer(g)(0,1+1) < \DDer(g)(0,1)+\DDer(g)(0,1)=4 $.
\end{remark}

The distance function on $Y^{X}$ in $\QQ$ can be characterized using derivatives as follows: given QLR $(X,Q,a)$ and $(Y,R,b)$ and functions $f,g\in Y^{X}$, let $(2, \{0<\infty\}, d_{\mathsf{disc}})$ be the QLR given by the discrete metric on $2=\{0,1\}$. Let $\mathbf h_{f,g}: 2\times X\to Y$ be the function given by 
$\mathbf h_{f,g}(0,x)=f(x)$ and $\mathbf  h_{f,g}(1,x)=g(x)$. A simple calculation yields then:
\begin{lemma}\label{lemma:distanceder}
$d^{\QQ}_{a,b}(f,g)(x,\alpha)=\DDer(\mathbf h_{f,g})(\langle\langle 0,x\rangle, \langle \infty, \alpha\rangle\rangle)$.
\end{lemma}
\longv{\begin{proof}
We have that
\begin{align*}
\DDer & (\mathbf h_{f,g})  (\langle\langle 0,x\rangle, \langle \infty, \alpha\rangle\rangle) \\
 &= \sup\{ b(\mathbf h_{f,g}(\langle 0,x\rangle), \mathbf h_{f,g}(\langle i,y\rangle))\mid   d_{\mathsf{disc}}(0,i)\leq \infty, a(x,y)\leq \alpha\}\\ 
 &=
\sup\{ b(f(x), f(y)), b(f(x),g(y))\mid a(x,y)\leq \alpha\}\\
&=
d^{\QQ}_{a,b}(f,g)(x,\alpha)
\end{align*}
\end{proof}
}

A consequence of Lemma~\ref{lemma:distanceder} is that the self-distance of $f\in Y^{X}$ coincides with its derivative, that is:
\begin{align}\label{eq:law}
d^{\QQ}_{a,b}(f,f)=\DDer(f)
\end{align}
Observe that this property implies that the self-distance of $f$ is (constantly) zero precisely when $f$ is a constant function.

We now define a category $\QQR$ of \emph{reflexive} QLR: $\QQR$ is the full subcategory of $\QQ$ made of QLR $(X,Q, a)$ such that $Q$ is Heyting and satisfies the property below:
\begin{align}\label{eq:heyting}
\text{if }\alpha \leq \beta \ \text{ then } \  \beta \leq  \beta \ResH \alpha \tag{$\star\star$}
\end{align}
and such that 
 $a(x,x)=0$ holds for all $x\in X$. 


The Lawvere quantale satisfies property \eqref{eq:heyting}, and this property is stable by 
product. In particular, 
$\QQR$ inherits the cartesian product from $\QQ$. 
Instead, the exponential of $(X,Q,a)$ and $(Y,R,b)$ in $\QQR$ is the QLR $(Y^{X}, R^{X\times Q}, d^{\QQR}_{a,b})$, where 
$$
d^{\QQR}_{a,b}(f,g):= d^{\QQ}_{a,b}(f,g) \ResH \DDer(f)
$$
Observe that $d^{\QQR}_{a,b}(f,f)= \DDer(f)\ResH\DDer(f)=0$.
The isomorphism $\begin{tikzcd}\QQR(Z\times X, Y)\ar[bend left=5]{r}{\lambda^{\mathsf r}} & \QQR(Z,Y^{X}) \ar[bend left=5]{l}[below]{\ev^{\mathsf r}}\end{tikzcd}$ is given by:
\begin{align*}
\lambda^{\mathsf{r}} (f,\varphi)& = (\lambda (f), \lambda(\varphi)\ResH \lambda z. \DDer( f (\langle z,\_))) \\
\ev^{\mathsf{r}}(f,\varphi)&= (\ev(f), \ev(\varphi) \vee \lambda z. \DDer(f(z)(\_)))
\end{align*}
%
%
\begin{remark}
In the absence of property \eqref{eq:heyting}, reflexive QLR only form a cartesian \emph{lax}-closed category \cite{Seely}. In particular, one has that $\ev^{\mathsf r}(\lambda^{\mathsf r}(f,\varphi))=\varphi$ and $\lambda^{\mathsf r}(\ev^{\mathsf r}(f,\psi))\leq \psi$ (in other words, $\beta$-reduction is preserved while $\eta$-reduction \emph{decreases} the interpretation).
\end{remark}

\begin{remark}\label{rem:distances}
In $\QQ$ and $\QQR$ we can define a ``na\"ive'' lifting of the Euclidean metric to all simple types built over the reals. This yields the two distance functions $d$ and $e$ on $\BB R^{\BB R}$ below:

\medskip
\noindent
\adjustbox{scale=0.9}{
\begin{minipage}{\linewidth}
\begin{align*}\label{eq:relaxed}
d(f,g)(x,\alpha)& =\sup\{d_{\mathsf{Euc}}(f(x),f(y)), d_{\mathsf{Euc}}(f(x),g(y))\shortv{\\ 
 &\qquad\qquad\qquad\qquad\qquad\qquad }\mid d_{\mathsf{Euc}}(x,y)\leq \alpha\} \\
e(f,g)(x,\alpha)&=
\begin{cases}
d(f,g)(x,\alpha) & \text{ if } d(f,g)(x,\alpha) > \DDer(f)(x,\alpha) \\
0 & \text{ otherwise}
\end{cases}
\end{align*}
\end{minipage}
}
\end{remark}

One can also consider categories $\QQS, \QQRS$ of \emph{symmetric} (resp.~reflexive and symmetric) QLR. 
One has the following:

\begin{lemma}\label{prop:symmetric0}
Let $(X,Q,a)$, $(Y,R,b)$ be symmetric QLR. If $R$ is a locale, then their exponential QLR in $\QQ$ is still symmetric.
\end{lemma}
\longv{
\begin{proof}[Proof of Lemma~\ref{prop:symmetric0}]
If $R$ is a locale, then we have that for all $x,y\in X$, $\alpha\in Q$ with $a(x,y)\leq \alpha$, 
$b(g(x),f(y)) \leq  b(g(x),f(x))\vee b(f(x),f(y))= b(f(x),f(y))\vee b(f(x),g(x))\leq d^{\QQ}_{a,b}(f,g)(x,\alpha)$ and 
$b(g(x),g(y)) \leq  b(g(x),f(x))\vee b(f(x),g(y))= b(f(x),g(x))\vee b(f(x),g(y))\leq d^{\QQ}_{a,b}(f,g)(x,\alpha)$, since $b$ is symmetric. From this we deduce that 
$d^{\QQ}_{a,b}(g,f)(x,\alpha)=
\sup\{ b(g(x),g(y)), b(g(x), f(y))\mid a(x,y)\leq \alpha\}\leq d^{\QQ}_{a,b}(f,g)(x,\alpha)$ and conversely.
\end{proof}
}

As a consequence, the categories $\QQU$ and $\QQRU$ of symmetric (resp.~reflexive and symmetric) QLR $(X,Q,a)$ where $Q$ is a locale, {are} cartesian closed subcategories of $\QQ, \QQR$, respectively.  We will meet these two categories in the next section.

The locale-valued symmetric QLR are essentially the only ones to inherit the cartesian closed structure of $\QQ$ and $\QQR$, as shown be the lemma below\longv{ (which is proved in the next section)}.

\begin{lemma}\label{prop:symmetric}
Let $(X,Q,a)$, $(Y,R,b)$ be symmetric QLR, where $Y$ is \emph{injective} (\cite{Espinola:2001aa,CLEMENTINO20063113}\longv{, see also Section \ref{section5}}) and $X$ contains two points $v_{0},v_{1}$ with $a(v_{0},v_{1})\neq 0$. Then, if the exponential of $X$ and $Y$ in $\QQ$ is symmetric, then for all $\alpha\in R$ such that $\alpha+\alpha\in Im(b)$, $\alpha=\alpha+\alpha$.
\end{lemma}



\subsection{QLR Models}

We now describe the interpretation of the simply typed $\lambda$-calculus inside $\QQ$ and $\QQR$. Concretely, this means associating each simple type with a QLR and each typed program with a morphism of QLR. We describe this situation abstractly through the notion of \emph{QLR-model}, introduced below.

%
%

%

\begin{definition}
Let $\BB C$ be a cartesian closed category. A \emph{$\QQ$-model} (resp. \emph{$\QQR$-model}) of $\BB C$ is a cartesian closed functor $F: \BB C\to \QQ$ (resp. $F:\BB C\to \QQR$).
\end{definition}

Concretely, a $\QQ$-model consists in the following data:
\begin{itemize}
\item for any object $X$ of $\BB C$, a QLR $(\model X, \nudel X, a_{X})$;
\item for any morphism $f\in \BB C(X,Y)$, functions $\model f: \model X\to \model Y$ and $\nudel f: \model X\times \nudel X\to \nudel Y$ such that $(\model f,\nudel f)$ is a QLR morphism from $\model X$ to $\model Y$, 
\end{itemize}
where the application $f\mapsto \nudel f$ satisfies suitable equations resembling Eq.~\ref{eq:prop2}-\ref{eq:prop7} (however, with equality in place of $\leq$). Observe that $\nudel f$ is in general only an \emph{approximation} of the derivative $\DDer{(\model f)}$ (that is, one has $\DDer{(\model f)}\leq \nudel f$).
%

%
%

We now describe a concrete $\QQ$-model for a simply typed $\lambda$-calculus $\STLC(\mathcal F)$ over a type $\Real$ for real numbers. More precisely, simple types are defined by the grammar
$$\sigma, \tau:=\Real\mid \sigma\to \tau\mid \sigma\times \tau$$
We fix a family $\C F=(\C F_{n})_{n>0}$ of sets of functions from $\BB R^{n}$ to $\BB R$. We consider the usual Curry-style simply-typed $\lambda$-calculus, with left and right projection $\pi_{1}$ and $\pi_{2}$, and with pair constructor $\langle \_,\_\rangle$, enriched with the following constants: for all $r\in \BB R$, a constant $\TT r:\Real$;
 for all $n>0$ and $f\in \C F_{n}$, a constant $\TT f: \Real \to \dots \to \Real \to \Real$.

The usual relation of $\beta$-reduction is enriched with the following rule, extended to all contexts: for all $n > 0$,  $f \in \C F_{n}$, and $ r_{1}, \dots,  r_{n} \in \BB R$, $\TT f\TT r_{1} \dots \TT r_{n} \longrightarrow_{\beta}\TT  s$, where $s = f(r_{1},...,r_{n})$.
  By standard arguments \cite{Krivine}, this calculus has the properties of subject reduction, confluence and strong normalization.

%
%
%
%

 We let $\Lambda(\C F)$ be the cartesian closed category whose objects are the simple types and where 
 $\Lambda(\C F)(\sigma, \tau)$ is the quotient of the set of closed terms of type $\sigma\to \tau$ under $\beta\eta$-equivalence, and composition of $[\lambda x.t]\in \Lambda(\C F)(\sigma, \tau)$ and 
 $[\lambda x.u]\in \Lambda(\C F)(\tau, \rho)$ is $[\lambda x.u(tx)]$.

 


A $\QQ$-model of $\STLC(\C F)$ is defined 
by setting $\model \Real =  \BB R$, $\nudel \Real =  \BB R^{\infty}_{\geq 0}$, $a_{\Real}=d_{\mathsf{Euc}}
$ and extending the definition of the QLR $(\model\sigma, \nudel \sigma, a_{\sigma})$ to all simple types $\sigma$ using the cartesian closed structure of $\QQ$.
%
Moreover, given a context $\Gamma=\{x_{1}:\sigma_{1},\dots, x_{n}:\sigma_{n}\}$ and a term $t$ of type $\Gamma \vdash t:\sigma$ (that we take as representative of a class of terms of type $(\prod_{i=1}^{n}\sigma_{i})\to \sigma$), the functions $\model t: \prod_{i=1}^{n}\model{\sigma_{i}} \to \model \sigma$ and $\nudel t: \prod_{i=1}^{n}\model{\sigma_{i}}\times \prod_{i=1}^{n}\nudel{\sigma_{i}}\to \nudel \sigma$ are defined by a straightforward induction on $t$. We unroll below the definition of $\nudel t$:
\begin{align*}
 \nudel{\TT r} (\vec x, \vec \alpha)& = 0  \\ 
  \nudel{\TT f}(\vec x, \vec \alpha) & = \DDer(f)(\vec x, \vec \alpha) \\
\nudel{x_{i}}(\vec x, \vec \alpha) & = \alpha_{i}\\
\nudel{\langle t,u\rangle}(\vec x, \vec \alpha)& =\langle \nudel t(\vec x, \vec \alpha),\nudel u(\vec x, \vec \alpha)\rangle\\
\nudel{t\pi_{i}}(\vec x, \vec \alpha) & = \pi_{i}(\nudel t(\vec x,\vec \alpha)) \\
\nudel{\lambda y.t}(\vec x, \vec \alpha) & = \lambda y\alpha.  \nudel t(\vec x*y, \vec \alpha*\alpha ) \\
\nudel{tu}(\vec x, \vec \alpha) & =
{\nudel t}(\vec x, \vec \alpha )(\model u(\vec x),\nudel u(\vec x,\vec \alpha))
\end{align*}
where $\vec x*y$ indicates the concatenation of $\vec x$ with $y$.

\begin{theorem}[Soundness]\label{thm:stlc}
For all simply typed terms $t$ such that $\Gamma \vdash t:\sigma$, $(\model t, \nudel t)\in \QQ(\model \Gamma, \model \sigma)$. Moreover, if $t\longrightarrow_{\beta}u$, then $\model t=\model u$ and $\nudel t=\nudel u$.
\end{theorem}

%


%

The following fact is an immediate consequence of Theorem \ref{thm:stlc} and Eq.~\eqref{eq:law}, and can be seen as a quantitative analog of the \emph{Fundamental Lemma} of logical relations, stating that any program $t$ is related to itself by $\nudel t$:

\begin{corollary}[Fundamental Lemma for QLR]\label{lemma:fundamental}
For all terms $t$ such that $\vdash t:\sigma $, 
$ a_{\sigma}(\model t, \model t)\leq \nudel t $.
\end{corollary}

Another quite literal consequence of Theorem \ref{thm:stlc} is that program distances are \emph{contextual}: given a distance between programs $t$ and $u$, for any context $\TT C[\_]$ we can obtain a distance between $\TT C[t]$ and $\TT C[u]$:

\begin{corollary}[contextuality of distances]\label{cor:context}
For all terms $t,u$ such that $\vdash t,u:\sigma$ holds and for all context $\TT C[\ ]: \sigma \vdash \tau$,
$$
a_{\tau}(\model{\TT C[t]}, \model{\TT C[u]}) \leq \nudel{\TT C}( \model t, a_{\sigma}(\model t, \model u))
$$
\end{corollary}

In a similar way one can define a $\QQR$-model of $\STLC(\C F)$ and prove analogs of the results above (where Corollary \ref{lemma:fundamental} now reads as
$a_{\sigma}(\model t, \model t)=0$).

%
%
%

\begin{remark}
Corollaries \ref{lemma:fundamental} and \ref{cor:context} generalize properties established in the setting of differential logical relations (cf.~Lemma 15 in \cite{dallago}).
\end{remark}

\begin{remark}
One can define an alternative interpretation of $\STLC$ by letting $\nudel t$ be the ``true'' derivative $\DDer(\model t)$. However, while Corollaries \ref{lemma:fundamental} and \ref{cor:context} still hold,
the operation $t\mapsto (\model t, \DDer(\model t))$ only yields a \emph{colax} functor (since one only has $\DDer(\model{u}\circ \model t)\leq \DDer(\model u)(\model t, \DDer(\model t))$).
%

\end{remark}

\section{Metrizability}\label{section5}

In this section we investigate generalized metrics in sub-categories of $\QQ$ and $\QQR$.
We first show that the relaxed and hyper-relaxed metrics all form cartesian closed subcategories of $\QQ$; we then turn to metrics and partial metrics: we show that, under suitable assumptions, the exponential QLR formed from two metric or partial metric spaces $X$ and $Y$ is a metric or a partial metric space precisely when the metric of $Y$ is idempotent (i.e.~distances satisfy $\alpha=\alpha+\alpha$).

%
%
%
%

This result can be used to show that ultra-metrics and partial ultra-metrics form cartesian closed subcategories of $\QQ$ and $\QQR$; at the same time it shows that the na\"ive lifting of the Euclidean metric (as well as of any non-idempotent metric) in either $\QQ$ or $\QQR$ is \emph{not} a generalized metric. 
Nevertheless, we show that liftings to all simple types can be defined for those metrics and partial metrics (including the Euclidean metric), whose
distance function factors as the composition of an idempotent metric and a \emph{valuation} \cite{ONeill, 10.1016/j.tcs.2003.11.016}.
%
%
%

%
%
%
%
%
%

\subsection{Relaxed metrics}

It is not difficult to check that whenever $(X,Q,a)$ and $(Y,R,b)$ are two relaxed or hyper-relaxed metrics, so is their exponential in $\QQ$. For the relaxed metrics, given $f,g,h\in Y^{X}$, using the triangular law of $Y$ we deduce that for all $x,y\in X$ and $\alpha\geq a(x,y)$, 
\begin{align*}
b(f(x), g(y)) & \leq b(f(x), h(x))+b(h(x),g(y)) \\
&\leq d^{\QQ}_{a,b}(f,h)(x, \alpha)+d^{\QQ}_{a,b}(h,g)(x,\alpha)
\end{align*} 
and thus that $d^{\QQ}_{a,b}(f,g)\leq d^{\QQ}_{a,b}(f,h)+d^{\QQ}_{a,b}(h,g)$.
This argument straightforwardly scales to the hyper-relaxed metrics, yielding:

\begin{proposition}\label{prop:relaxed}
The full subcategories of $\QQ$ made of relaxed and hyper-relaxed metrics are cartesian closed.
\end{proposition}

An immediate consequence is that the distance $d$ from Remark \ref{rem:distances} is a relaxed metric. 
We will show below that we cannot actually say \emph{more} of $d$: it is not a partial metric.

\subsection{Ultra-metrics}

For all metric spaces $(X,Q,a)$ and $(Y,R,b)$, whenever $R$ satisfies $\alpha+\beta=\alpha\vee \beta$ (or, equivalently, $\alpha=\alpha+\alpha$ and $0=\bot$), it is not difficult to check that the transitivity axiom lifts to the exponential in $\QQ$: in fact, for all $f,g,h\in Y^{X}$ and $x,y\in X$ with $a(x,y)\leq \alpha$ one has 
\begin{align*}
b(f(x),g(y)) & \leq b(f(x),h(x)) \vee b(h(x),g(y)) \\
& \leq d^{\QQ}_{a,b}(f,h)(x,\alpha) \vee d^{\QQ}_{a,b}(h,g)(x,\alpha)
\end{align*}
from which we deduce $d^{\QQ}_{a,b}(f,g)(x,\alpha)\leq d^{\QQ}_{a,b}(f,h)(x,\alpha) \vee d^{\QQ}_{a,b}(h,g)(x,\alpha)$.
A similar argument can be developed for the distance $d^{\QQR}_{a,b}$, leading to:

\begin{proposition}\label{prop:ultra}
The full subcategories of $ \QQRU$ and $\QQU$ made of ultra-metric spaces and partial ultra-metric spaces are cartesian closed.
\end{proposition}
\longv{\begin{proof}
Let $(X,Q,a), (Y,R,b)$ be objects of $\QQRU$. It suffices to show that the QLR $Y^{X}$ satisfies transitivity. Since $R$ is a locale, $\alpha+ \beta=\alpha\vee \beta$ holds for all $\alpha,\beta\in R$. Let $f,g,h\in Y^{X}$. Then we have that 
$\DDer(f)\vee (d^{\QQ}_{a,b}(f,h)+d^{\QQ}_{a,b}(h,g))=(\DDer(f)\vee d^{\QQ}_{a,b}(f,h))\vee d^{\QQ}_{a,b}(h,g)$, so in particular for all $x,y\in X$ and $\alpha\geq a(x,y)$, 
$(\DDer(f)\vee (d^{\QQ}_{a,b}(f,g)+d^{\QQ}_{a,b}(h,g)))(x,\alpha) = ((\DDer(f)\vee d^{\QQ}_{a,b}(f,g))(x,\alpha)) \vee((\DDer(f)\vee d^{\QQ}_{a,b}(h,g))(x,\alpha))\geq 
d^{\QQ}_{a,b}(f(x), h(x))\vee d^{\QQ}_{a,b}(h(x),g(y)) \geq d^{\QQ}_{a,b}(f(x),g(y))$, from which we deduce that 
$(d^{\QQ}_{a,b}(f,g)+d^{\QQ}_{a,b}(h,g))(x,\alpha) \ResH \DDer(f)(x,\alpha)\geq( d^{\QQ}_{a,b}(f,g)(x,\alpha))\ResH (\DDer(f) (x,\alpha))=d^{\QQR}_{a,b}(f,g)(x,\alpha)$.

A similar argument can be developed for $\QQU$, using the fact that in a locale $\alpha+_{\gamma}\beta=\alpha\vee \beta$.
\end{proof}
}

When $Q$ is a locale, also the category $\Met_{Q}$ is cartesian closed \cite{Smyth}.
These categories have been mostly used to account for \emph{intensional} properties of higher-order programs (e.g.~measuring program approximations or the number of computation steps \cite{Escardo1999}).
In the categories $\QQRU$ and $\QQU$ we can define metrics describing more \emph{extensional} properties (i.e.~measuring distances between program outputs) as the one below.
\begin{example}\label{ex:intervals}
Let $\C I(\BB R)$ be the complete lattice of \emph{closed intervals} $[x,y]$ (where $x,y\in \BB R$ and $x\leq y$), enriched with $\emptyset$ and $\BB R$. We can define a partial ultra-metric $u: \BB R\times \BB R\to \C I(\BB R)$ by letting $u(x,y)=[\min\{x,y\}, \max\{x,y\}]$. 

The metric $u$ lifts in $\QQU$ to a partial ultra-metric $d^{\QQ}_{u,u}$ over real-valued functions where, for all $x\in \BB R$ and $I\in \C I(\BB R)$, $d^{\QQ}_{u,u}(f,g)(x,I)$ is the smallest interval containing all $f(y)$ and $g(y)$, for $y\in I\vee\{x\}$ (see also \cite{Geoffroy2020}). 
\end{example}

We now establish a sort of converse to Proposition \ref{prop:ultra}: under suitable conditions, if the exponential of two metric spaces $X$ and $Y$ satisfies the transitivity axiom, then the distances over $Y$ are idempotent\shortv{:}\longv{.

Let us first recall the notion of \emph{injective} metric space \cite{Espinola:2001aa,CLEMENTINO20063113}, that will be essential in our argument.
A map $f:X\to X$ between two metric spaces $(X,Q,a), (Y,Q,b)$ over the \emph{same} quantale is said an \emph{extension} if for all $x,y\in X$, $b(f(x),f(y))=a(x,y)$, and is said \emph{non-expansive} if for all $x,y\in X$, $b(f(x),f(y))\leq a(x,y)$.

A metric space $(X,Q,a)$ is \emph{injective} when for all non-expansive map $f: Y\to X$ and extension $e:Y\to Z$ there exists a non-expansive map $h:Z\to X$ such that $f=h\circ e$.

Injective metric spaces (also known as \emph{hyperconvex} metric spaces) enjoy several nice properties (see \cite{Espinola:2001aa}). In particular, they form a cartesian closed subcategory of $\Met$ \cite{CLEMENTINO20063113}, which includes the Euclidean metric. 
Here we will use such spaces to establish a few negative results.}

\begin{lemma}\label{lemma:trivialmetric}
\begin{itemize}
\item[i.]Let $(X,Q,a)$ and $(Y,R,b)$ be two metric spaces, where $X$ has at least two distinct points and $Y$ is \emph{injective}\shortv{ (\cite{Espinola:2001aa,CLEMENTINO20063113})}. 
If the reflexive QLR $(Y^{X}, R^{X\times Q}, d^{\QQR}_{a,b})$ is a metric space then for all $\alpha,\beta\in R$ such that $\alpha+\beta\in Im(b)$, $\alpha+\beta=\alpha\vee \beta$.
\item[ii.] Let $(X,Q,a)$ and $(Y,R,b)$ be two partial metric spaces, where $X$ has at least two distinct points and $Y$ is injective. 
If the QLR $(Y^{X}, R^{X\times Q}, d^{\QQ}_{a,b})$ is a partial metric space then for all $\alpha,\beta\in R$ such that $\alpha+\beta\in Im(b)$, $\alpha+\beta=\alpha\vee \beta$.
\end{itemize}
\end{lemma}
\longv{
\begin{proof}
\begin{itemize}
\item[i.]
Let $\alpha,\beta\in R$ and $u_{0},u_{2}\in Y$ be such that $b(u_{0},u_{2})=\alpha+\beta$.
Let $Y'=Y\cup \{v_{1}\}$ and $b'$ be as $b$ on $Y$ and satisfying $b(u_{0},v_{1})=\alpha$, 
$b(v_{1},u_{2})=\beta$. The injection $\iota:Y\to Y'$ is an expansion, hence, since $Y$ is injective, there exists a non-expansive function $f: Y'\to Y$ such that $f\circ \iota= \mathrm{id}_{Y}$. This implies in particular that, by letting $u_{1}:= f(v_{1})$, $b(u_{0},u_{1})\leq\alpha$, $b(u_{1},u_{2})\leq \beta$.

Let now $x_{0},x_{1}$ be two distinct points in $X$ and let 
$f,g,h:X\to Y$ be the following functions: 
$f(x)$ is constantly $u_{0}$ except for $f(x_{1})=u_{1}$; $g(x)$ is constantly $u_{2}$ and $h(x)$ is constantly $u_{1}$. 
We have then that $\DDer(f)(x,a(x_{0},x_{1}))\leq \alpha$, $\DDer(g)=\DDer(h)=0$. Moreover, 
for all $x'\in X$ with 
 $a(x_{0},x')\leq a(x_{0},x_{1})$, $b(f(x_{0}),f(x')), b(f(x_{0}),h(x'))\leq b(u_{0},u_{1})\leq \DDer(f)(x,a(x_{0},x_{1}))=\DDer(f)(x_{0},a(x_{0},x_{1})) \vee 0
$, that is $d^{\QQ}_{a,b}(f,h)(x_{0},a(x_{0},x_{1}))\leq \DDer(f)(x,a(x_{0},x_{1}))$, and thus \\ $d^{\QQR}_{a,b}(f,h)(x,a(x_{0},x_{1}))=d^{\QQ}_{a,b}(f,h)\ResH \DDer(f))(x_{0},a(x_{0},x_{1}))\leq 0$


Then, since by hypothesis $e_{a,b}$ is a metric, we deduce that 
\begin{align*}
\alpha+\beta& =
b(u_{0},u_{2})=b(f(x_{0}),g(x_{1}))  \\
&
\leq d^{\QQ}_{a,b}(f,g)(x_{0},a(x_{0},x_{1})) \\
&
\leq \big (\DDer(f)\vee d^{\QQR}_{a,b}(f,g)\big )(x_{0},a(x_{0},x_{1})) \\
&
\leq\big (\DDer(f)\vee(d^{\QQR}_{a,b}(f,h)+d^{\QQR}_{a,b}(h,g))\big )(x_{0},a(x_{0},x_{1}))\\
&\leq 
\alpha \vee (0+\beta)= \alpha\vee \beta
\end{align*}

\item[ii.]
As in the proof of point i.~let $\alpha,\beta\in R$ and $u_{0},u_{1},u_{2}\in Y$ be such that $b(u_{0},u_{1})\leq \alpha$, $b(u_{1},u_{2})\leq \beta$ and 
$b(u_{0},u_{2})=\alpha+\beta$. We can suppose w.l.o.g. that $b$ is symmetric.

Let now $x_{0},x_{1}$ be two distinct points in $X$ and let 
$f,g,h:X\to Y$ be the following functions: 
$f(x)$ is constantly $u_{0}$, $h(x)$ is constantly $u_{1}$ except for $h(x_{1})=u_{0}$ and $g(x)$ is constantly $u_{1}$ except for $g(x_{1})=u_{2}$. 
Then we have that $d^{\QQ}_{a,b}(f,g)(x_{0},a(x_{0},x_{1}))=b(u_{0},u_{2})=\alpha+\beta$, 
$d^{\QQ}_{a,b}(f,h)(x_{0},a(x_{0},x_{1}))=d^{\QQ}_{a,b}(h,h)(x_{0},a(x_{0},x_{1}))=b(u_{0},u_{1})\leq \alpha$ and 
$d^{\QQ}_{a,b}(h,g)(x_{0},a(x_{0},x_{1}))= b(u_{0},u_{1})\vee b(u_{1},u_{2})\leq \alpha \vee \beta$. 

Then, since by hypothesis $d^{\QQ}_{a,b}$ is a partial metric, we deduce that 

\noindent
\adjustbox{scale=0.95}{
\begin{minipage}{\linewidth}
\begin{align*}
\alpha+\beta& =
b(u_{0},u_{2})=b(f(x_{0}),g(x_{1}))  \\
& \leq d^{\QQ}_{a,b}(f,g)(x_{0},a(x_{0},x_{1}))\\
 & \leq \big ((d^{\QQ}_{a,b}(f,h)\Res d^{\QQ}_{a,b}(h,h))+d^{\QQ}_{a,b}(h,g)\big )(x_{0},a(x_{0},x_{1}))
\\
&
= \big ((d^{\QQ}_{a,b}(h,h)\Res d^{\QQ}_{a,b}(h,h))+d^{\QQ}_{a,b}(h,g)\big )(x_{0},a(x_{0},x_{1})) \\
&
= d^{\QQ}_{a,b}(h,g)(x_{0},a(x_{0},x_{1}))
\leq 
\alpha \vee \beta
\end{align*}
\end{minipage}
}
\end{itemize}
\end{proof}

}


To give the reader an \shortv{idea of the proof}\longv{illustration} of Lemma \ref{lemma:trivialmetric}, we \shortv{illustrate}\longv{show} in Fig.~\ref{fig:counterexamples} counter-examples to transitivity for the na\"ive extensions of the Euclidean metric (cf.~Remark \ref{rem:distances}).

\longv{
Along similar lines we can also prove Lemma \ref{prop:symmetric} from the previous section.

\begin{proof}[Proof of Lemma~\ref{prop:symmetric}]

Let $\alpha\in R$ and $x_{1},x_{2}\in Y$ be such that $b(x_{1},x_{2})=\alpha+\alpha$. Let $(Z,R,c)$ be a metric space where  $Z=X\cup\{u_{0},u_{3}\}$ and $c$ is defined so that $c(u_{0},u_{0})=c(u_{3},u_{3})=0$ and the following hold:
\begin{align*}
c(u_{0},u_{1}), c(u_{0},u_{2}), c(u_{0},u_{3})& =\alpha \\
c(u_{1},u_{2}), c(u_{2},u_{3}), c(u_{3},u_{1}) & = \alpha+\alpha
\end{align*}
Since $Y$ is injective, there exists a non-expansive map $f:Y\to X$ such that $f\circ \iota=\mathrm{id}_{X}$, where $\iota$ is the injection $\iota: X\to Z$ (which is obviously an expansion). Hence there exist points $x_{0},x_{3}\in X$ such that 
$ b(x_{0},x_{1}), b(x_{0},x_{2}), b(x_{0},x_{3}) =\alpha$ and 
$b(x_{1},x_{2}), b(x_{2},x_{3}), b(x_{3},x_{1}) \leq \alpha+\alpha$. 

Let $f,g\in Y^{X}$ be defined by 
\begin{align*}
f(w)=
\begin{cases}
x_{1} & \text{ if } w=v_{0} \\
x_{2} & \text{ otherwise}
\end{cases}
\qquad
g(w)=
\begin{cases}
x_{3} & \text{ if } w=v_{0} \\
x_{4} & \text{ otherwise}
\end{cases}
\end{align*}
where $v_{0},v_{1}$ are two distinct points of $X$ such that $a(v_{0},v_{1})\neq 0$. 
%
%
%
If $d^{\QQ}_{a,b}(f,g)=d^{\QQ}_{a,b}(g,f)$, we deduce that 
\begin{align*}
\alpha & \geq
 \sup\{ b(f(v_{0}),f(w)), b(f(v_{0}),g(w))\mid a(v_{0},w)\leq a(v_{0},v_{1})\}\\
 & =
d^{\QQ}_{a,b}(f,g)(v_{0}, a(v_{0},v_{1}))\\
& =d^{\QQ}_{a,b}(g,f)(v_{0},a(v_{0},v_{1}))
\\ & =
 \sup\{ b(g(v_{0}),g(w)),b(g(v_{0}),f(w))\mid a(v_{0},w)\leq a(v_{0},v_{1})\}\\
 & =
 \alpha+\alpha
 \end{align*}
\end{proof}

}


%

\begin{figure*}
\fbox{
\begin{subfigure}{0.47\textwidth}
\begin{center}
\begin{tikzpicture}[domain=-2:2]
\draw[dashed, |<->|]   (-2,0) -- (2,0);
\draw[<->] (0,-0.4) -- (0,3); 

    \node at (0,0)[circle,fill,inner sep=1pt]{};
\node(z) at (0.2,-0.2) {$x$};
\node(-e) at (-2,-0.2) {$x-r$};
\node(e) at (2,-0.2) {$x+r$};

\node(f) at (0,2.6)[circle,fill,inner sep=1pt]{};
\node(f) at (0,1.45)[circle,fill,inner sep=1pt]{};
\node(f) at (0,0.3)[circle,fill,inner sep=1pt]{};

\node(f) at (-0.2,2.4) {\tiny$g(x)$};

\node(f) at (-0.2,1.6) {\tiny$h(x)$};

\node(f) at (-0.2,0.5) {\tiny $f(x)$};

\node(ff) at (1.3,2.3) {{$g$}};
\node(gg) at (1.2,1.1) {{$h$}};
\node(gg) at (1,0.45) {{$f$}};

\draw[color=orange] plot (\x,{2.25+0.8*(-cos(\x r))}) ;
\draw[color=red] plot (\x,{0.65+0.8*(cos(\x r))}) ;
\draw[color=blue] plot (\x,{0.3}) ;

\draw[dotted] (0,0.3) -- (2.4,0.3);
\draw[dotted] (0,1.45) -- (2.4,1.45);
\draw[dotted] (0,2.6) -- (2.4,2.6);

%

\draw[dashed,|<->| ] (2.4,2.6) -- node[right] {\tiny$d(h,g)$} (2.4,1.45);
\draw[dashed,|<->| ] (2.4,1.45) -- node[right] {\tiny$d(f,h)=d(h,h)$} (2.4,0.3);
\draw[dashed,|<->| ] (2.8,2.6) -- node[right] {\tiny$d(f,g)$} (2.8,0.3);




\end{tikzpicture}
\end{center}
\caption{\small The distance $d$ from Remark \ref{rem:distances} is not a partial metric. 
 The example above shows that $d(f,g)> d(f,h)+d(h,g)- d(h,h)$ (with all distances computed in $(x,r)$). A similar example can be found in~\cite{Geoffroy2020}. \\ \ \longv{\\ \ } }
\end{subfigure}
\ \ \ 
\begin{subfigure}{0.47\textwidth}
\begin{center}
\begin{tikzpicture}[domain=-2:2]
\draw[dashed, |<->|]   (-2,0) -- (2,0);
\draw[<->] (0,-0.4) -- (0,3); 

    \node at (0,0)[circle,fill,inner sep=1pt]{};
\node(z) at (0.2,-0.2) {$x$};
\node(-e) at (-2,-0.2) {$x-r$};
\node(e) at (2,-0.2) {$x+r$};

\node(f) at (0,2.6)[circle,fill,inner sep=1pt]{};
\node(f) at (0,1.45)[circle,fill,inner sep=1pt]{};
\node(f) at (0,0.3)[circle,fill,inner sep=1pt]{};

\node(f) at (-0.2,2.4) {\tiny$g(x)$};

\node(f) at (-0.2,1.6) {\tiny$h(x)$};

\node(f) at (-0.2,0.5) {\tiny $f(x)$};

\node(ff) at (1.3,2.3) {{$g$}};
\node(gg) at (1.2,1.1) {{$h$}};
\node(gg) at (1,0.45) {{$f$}};

\draw[color=blue] plot (\x,{1.1+0.8*(-cos(\x r))}) ;
\draw[color=orange] plot (\x,{1.8+0.8*(cos(\x r))}) ;
\draw[color=red] plot (\x,{1.45}) ;

\draw[dotted] (0,0.3) -- (2.4,0.3);
\draw[dotted] (0,1.45) -- (2.4,1.45);
\draw[dotted] (0,2.6) -- (2.4,2.6);

%

\draw[dashed,|<->| ] (2.4,2.6) -- node[right] {\tiny$d(h,g)$} (2.4,1.45);
\draw[dashed,|<->| ] (2.4,1.45) -- node[right] {\tiny$d(f,f)=d(f,h)$} (2.4,0.3);
\draw[dashed,|<->| ] (2.8,2.6) -- node[right] {\tiny$d(f,g)$} (2.8,0.3);




\end{tikzpicture}
\end{center}
\caption{\small The distance $e$ from Remark \ref{rem:distances} is not a metric. 
In the example above (with all values computed in $(x,r)$), $e(f,h)=0$, since each $h(y)$ is no farther from $f(x)$ than $f(x+r)$, $e(h,g)=d(h,g)$ and $e(f,g)$ is $d(f,f)+d(f,g)$. Hence transitivity fails since 
 $e(f,g)= d(f,h)+d(h,g)  >
0+d(h,g)= 
e(f,h)+e(h,g)$.}
\end{subfigure}
}
\caption{The distances $d$ and $e$ from Remark \ref{rem:distances} do not satisfy the transitivity axioms of metric and partial metric spaces.}
\label{fig:counterexamples}
\end{figure*}
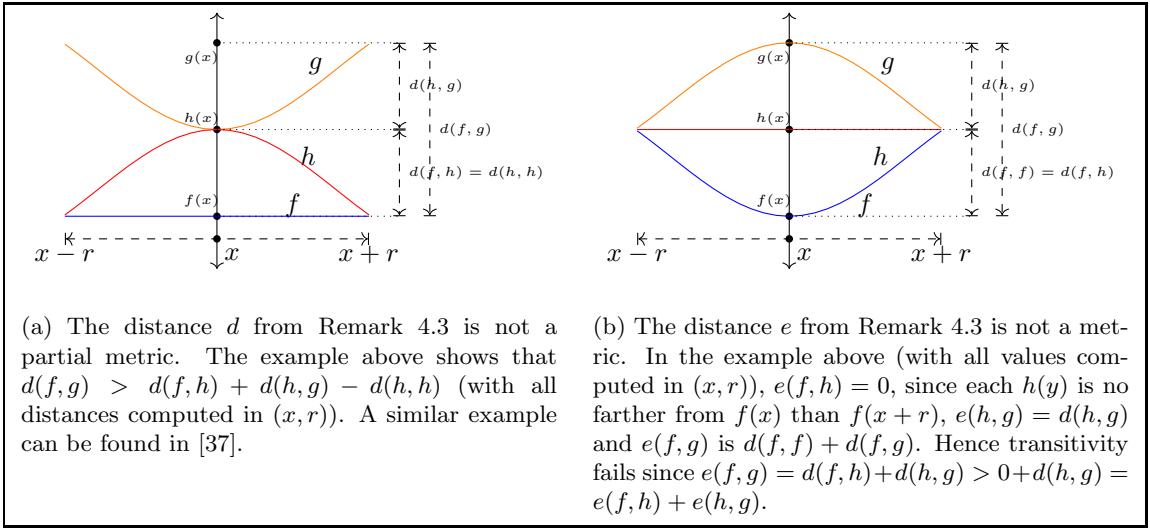

\subsection{Decomposing Partial Metrics through Valuations}

Lemma~\ref{lemma:trivialmetric} suggests that one cannot hope to lift the Euclidean metric to all simple types inside $\QQ$ or $\QQR$. 
Nevertheless, we will show that the Euclidean metric, as well as many other non-idempotent metrics and partial metrics, can be lifted to all simple types inside the categories $\QQU$ and $\QQRU$, by exploiting a well-investigated connection between partial metrics and lattice-valued metrics.

A basic intuition comes from the observation that the Euclidean distance can be \emph{decomposed} as
$$
\begin{tikzcd}
\BB R\times \BB R \ar{r}{u}  & \C I(\BB R) \ar{r}{\mu} &  \BB R^{+}_{\geq 0} 
\end{tikzcd}
$$
where $u$ is the partial ultra-metric from Example \ref{ex:intervals} and $\mu$ is the Lebesgue measure. 
This observation can be generalized using the theory of \emph{valuations} \cite{ONeill, 10.1007/BFb0053546,10.1016/j.tcs.2003.11.016}.


A \emph{join-valuation} \cite{10.1016/j.tcs.2003.11.016} on a join semi-lattice $L$ is a monotone function $\mathcal F: L\to \BB R^{+\infty}_{\geq 0}$ which satisfies the condition 
\begin{equation}\label{eq:submodular}
\C F( a \lor b) \leq \C F(a) + \C F(b) -\C F(a\land b)
\end{equation}
for all $a,b$ such that  $a\land b$ exists in $L$. When $L$ is a $\sigma$-algebra, join-valuations on $L$ are thus sort of relaxed measures on $L$.

Any join-valuation $\C F:L\to \BB R^{+\infty}_{\geq 0}$ induces a join semi-lattice $L_{\C F}$ obtained by quotienting $L$ under the equivalence  
$$
a \simeq_{\C F}b \ \text{iff} \ (a\leq b \text{ or } b\leq a) \text{ and }\C F(a)=\C F(b) 
$$
One can obtain then a separated and symmetric partial metric $p_{\C F}:L_{\C F}\times L_{\C F}\to \BB R^{+\infty}_{\geq 0}$ by letting $p_{\C F}(a,b)=\C F(a\vee b)$. The transitivity axiom is checked as follows:
\begin{align*}
\C F(a \lor b) & \leq\C  F((a\lor c)\lor (c\lor b)) \\
& \leq\C  F(a\lor c)+ \C F(c\lor b)- \C F((a\lor c) \land (c\lor b)) \\
& \leq \C F(a\lor c)+ \C F(c\lor b)- \C F(c\lor c)
\end{align*}

\begin{remark}
The connection between partial metrics and valuations has a converse side \cite{10.1016/j.tcs.2003.11.016}: any (symmetric and separated) partial metric $p:X\times X\to \BB R^{+\infty}_{\geq 0}$ defines an order $\sqsubseteq_{p}$ over $X$ given by $x\sqsubseteq_{p}y $ iff $p(x,y)\leq p(x,x)$. Then, whenever the poset $(X, \sqsubseteq_{p})$ is a join semi-lattice, the self-distance function $ X \stackrel{\Delta}{\to} X\times X \stackrel{p}{\to} \BB R^{+\infty}_{\geq 0}$ is a join-valuation.
\end{remark}

Extending this observation to arbitrary (commutative and integral) quantales leads to the following:

\begin{definition}
A \emph{(generalized) valuation space} (noted $\COV{L}{\C F}{Q}$) is the given of a monotone function from a complete lattice $L$ to a quantale $Q$ satisfying
\begin{equation}
\C F(a\lor b) \leq \C F(a) + (\C F(b) \Res\C  F(a\land b))
\end{equation}
for all $a,b\in L$ such that $a\land b\neq \bot$.
\end{definition}

By arguing as above, any valuation space $\COV{L}{\C F}{Q}$ yields a (symmetric and separated) partial metric $\C F: L_{\C F}\times L_{\C F}\to Q$.
%
%
This leads to the following definition:

\begin{definition}
A \emph{partial metric valuation space} is a triple $(X, \COV{L}{\C F}{Q}, a)$, where $ \COV{L}{\C F}{Q}$ is a valuation space and
$UX=(X, L_{\C F},a)$ is a (symmetric and separated) partial ultra-metric space.

A map of partial metric valuation spaces $(X, \COV{L}{\C F}{Q},a)$ and $(Y,\COV{M}{\C G}{R},b)$ is an arrow $(f,\varphi)$ in 
$\QQU( UX, UY)$.
%
\end{definition}

Observe that any partial metric valuation space $(X, \COV{L}{\C F}{Q}, a)$ yields \emph{both} a partial ultra-metric $a:X\times X\to L_{\C F}$ and a (separated) partial metric $\C F\circ a:X\times X\to Q$.

\begin{example}
The Euclidean metric can be presented as a partial metric valuation space in two ways: either using the Lebesgue measure as shown before, or
by considering the valuation space $\COV{\C I(\BB R)^{-}}{\mathsf{diam}}{\BB R^{+\infty}_{\geq 0}}$ where $\C I(\BB R)^{-}$ is the join-semilattice $\C I(\BB R)-\{\emptyset\}$ and $\mathsf{diam}$ is the diameter function (which is in fact \emph{modular} over intersecting intervals, see \cite{Geoffroy2020}).
\end{example}
%
%
Observe that for any map $(f,\varphi)$ of spaces $(X, \COV{L}{\C F}{Q},a)$ and $(Y, \COV{M}{\C G}{R},b)$, we have that for all $x,y\in X$ and $\alpha\in L$, 
\begin{align*}
\C G(b(f(x),f(y))\leq \C G(\varphi(x,\alpha))
\end{align*}
In other words, the composition of derivatives and valuations provides a compositional way to compute distance bounds.

We let $\mathsf p\VPM$ indicate the category of partial metric valuation spaces.
Since the functor $U: \mathsf p\VPM\to \QQU$ is by definition full and faithful, $\mathsf p\VPM$ inherits the cartesian closed structure from $\QQU$.
In particular, given partial metric valuation spaces $(X, \COV{L}{\C F}{Q},a)$ and $(Y, \COV{M}{\C G}{R},b)$, their product and exponential are as follows: 
\begin{center}
$
(X\times Y, \begin{tikzcd}L\times M\ar{r}{\C F\times\C G} & R\times Q\end{tikzcd}, a\times b)
$\\
$
(Y^{X}, \begin{tikzcd} (R_{\C G})^{X\times L_{\C F}}\ar{r}{ \C G\circ \_  }& Q^{X\times L_{\C F}}\end{tikzcd} , d_{a,b})
$
\end{center}

\begin{example}\label{ex:pmet}
The exponential of the Euclidean metric in $\mathsf p\VPM$ is the partial metric $p: (\BB R^{\BB R}\times \BB R^{\BB R}) \to (\BB R^{+}_{\geq 0})^{\BB R\times \C I(\BB R)}$ given by 
$$
p(f,g)(x, I)=\mathsf{diam}\{ b(f(y), g(z))\mid y,z\in \{x\}\vee I\}
$$
We can compare $p$ with the na\"ive lifting $d$ in Fig.~\ref{fig:counterexamples}, by considering the interval $I=[x-r,x+r]$. One has 
$p(f,h)(x,I)=d(f,h)$ but $p(h,g)(x,I)= d(h,h)+d(h,g)$. Hence transitivity holds for $p$, since 
$p(f,g)(x,I)= p(f,h)(x,I)+ p(h,g)(x,I)-p(h,h)(x,I)$. 
\end{example}

This construction can be adapted to metric spaces. Let a \emph{dual join-valuation} be a monotone map
$L^{\mathsf{op}}\times L \stackrel{\C D}{\to} Q$ (where $L^{\mathsf{op}}$ is the complete lattice with the reversed order) 
satisfying 
\begin{align*}
\C D(a,a) =0\qquad \qquad
\C D(a,b\vee c) \leq \C D(a, b)+ \C D(b\land c, c)
\end{align*}
One defines the quotient $L_{\C D}$ by $a\simeq_{\C D}b$ iff $a\leq b$ or $b\leq a$ and $\C D(a,a\vee b)=\C D(b,b\vee a)=0$. 
For any dual join valuation $\C D$, the function $d_{\C D}:L_{\C D}\times L_{\C D}\to Q$ given by $d(a,b)=\C D(a,a\vee b)+\C D(b,b\vee a)$ is a symmetric and separated metric. Moreover,
any join-valuation $L\stackrel{\C F}{\to} Q$ yields the dual join valuation $\C F'(a,b)=\C F(b)\Res\C  F(a)$.

Let a \emph{metric valuation space} be a triple $(X, L^{\mathsf{op}}\times L \stackrel{\C D}{\to} Q, a)$, where $L^{\mathsf{op}}\times L \stackrel{\C D}{\to} Q$ is a dual join valuation and $UX=(X,L_{\C D}, a)$ is a symmetric and separated ultra-metric space. 
One obtains then a category $\VPM$ of metric valuation spaces, with $\VPM(X,Y)=\QQRU(UX,UY)$.

\begin{theorem}
The categories $\mathsf p\VPM$ and $\VPM$ are cartesian closed.
\end{theorem}

\begin{example}
The Euclidean metric lives in $\VPM$ as it arises from the
dual join valuation $\C D: \C I(\BB R)^{\mathsf{op}}\times \C I(\BB R)\to \BB R^{\infty}_{\geq 0}$ given 
by $\C D(I,J)=\mathsf{diam}(J)\Res \mathsf{diam}(I)$. Its lifting to $\BB R^{\BB R}$ inside $\VPM$ yields the metric $m(f,g)=2p(f,g)-p(f,f)-p(g,g)$,
where $p$ is the partial metric from Example \ref{ex:pmet}. 
\end{example}
%


\section{A Generalized Lipschitz Condition}\label{section6}


In this section we explore a different class of morphisms between QLR, generalizing the usual Lipschitz condition. Notably, we show that in this setting the QLR satisfying reflexivity and transitivity can be lifted to all simple types.



\subsection{From Lipschitz to Locally Lipschitz functions}

%
%

%
%

As observed in previous sections, the Lipschitz condition has been widely investigated in program semantics, 
but is considered problematic when dealing with fully higher-order languages.
Does the picture change when we step from models like $\Met_{Q}$ to categories of QLR?

\begin{remark}
For simplicity, from now on we will suppose that QLR are always reflexive and symmetric.
\end{remark}

To answer this question we must first find a suitable extension of the Lipschitz condition to this setting. 
The first step is to  introduce a notion of \emph{finiteness}: since a quantale is a complete lattice, we must avoid that \emph{any} function $f:X\to Y$ between QLR admits the trivial Lipschitz constant $\top$. 
\begin{definition}
Let $Q$ be a commutative and integral quantale. A \emph{finiteness filter of $Q$} is a downward set $\ff Q\subseteq Q$ such that $a,b\in \ff Q$ implies $a+b\in \ff Q$.

A \emph{finitary QRL} is a tuple $(X,Q,\ff Q,a)$ such that $(X,Q,a)$ is a QLR, $\ff Q$ is a finitary filter of $Q$ and $Im(a)\subseteq \ff Q$.


\end{definition}

The positive reals $\BB R_{\geq 0}$ form a finiteness filter of $\BB R^{\infty}_{\geq 0}$. Moreover, if $\ff Q$ and $\ff R$ are finiteness filters of $Q$ and $R$, then $\ff Q\times \ff R$ is a finiteness filter of $Q\times R$, and 
for all set $X$, $(\ff Q)^{X}$ is a finiteness filter of $Q^{X}$.

Now, a basic observation is that if a function $f:X\to Y$ between metric spaces is $L$-Lipschitz, then there is a monoid homomorphism $\varphi: \BB R^{+}_{\geq 0}\to \BB R^{+}_{\geq 0}$ given by $\varphi(x)=L\cdot x$, such that $d(f(x),f(y))\leq \varphi(d(x,y))$. 
 This suggests the following:
 \begin{definition}[generalized Lipschitz maps]
Let $(X,Q,\ff Q, a)$, $(Y,R,\ff R,b)$ be finitary QLR. A function $f:X\to Y$ is a \emph{generalized Lipschitz map} from $X$ to $Y$ if there exists a monoid homomorphism $\varphi:Q\to R$ satisfying:
\begin{align*}
\forall \alpha\in \ff Q \ &  \varphi(\alpha)\in \ff R \tag{finiteness} \\
a(x,y)\leq \alpha \ & \To \  b(f(x),f(y))\leq \varphi(\alpha) \tag{Lipschitz}
\end{align*}
\end{definition}

Observe that any Lipschitz function $f: \BB R\to \BB R$ in the usual sense is a generalized Lipschitz map between the finitary and reflexive QLR given by the Euclidean metric.

Finitary QLR and generalized Lipschitz maps form a category $\mathbf L$ with cartesian structure defined as in $\QQ$. Moreover, given finitary QLR $(X,Q,\ff Q,a)$ and $(Y,R,\ff R,b)$ there is a finitary QLR $(\mathbf L(X,Y), R^{X},(\ff R)^{X}, b^{X})$ where $b^{X}(f,g)(x)=b(f(x),g(x))$ (note that also symmetry and reflexivity are preserved).

Yet, with this definition $\mathbf L$ is still \emph{not} cartesian closed. 
For instance, consider the function $f(x)(y): \BB R\to \BB R^{\BB R}$ given by $f(x)(y)=x \cdot y$. As a function of two variables, $f$ is Lipschitz in both $x$ and $y$, with Lipschitz constants $|y|$ and $|x|$; one can use this fact to show that $f\in \ELLE (\BB R, \BB R^{\BB R})$.
%
Now, if $\ELLE$ were cartesian closed, using the canonical isomorphism $\ELLE(\BB R, \BB R^{\BB R})\simeq \ELLE(\BB R\times \BB R, \BB R)$, we could deduce that also the function $\ev(f)(\langle x,y\rangle)=f(x)(y)$ is Lipschitz. 
However, there is no way to deduce, from the two piecewise Lipschitz constants $|x|$ and $|y|$ for $x$ and $y$, a uniform Lipschitz constant for \emph{both} variables. 
In fact, all we can say is that, for any choice of points $x,y\in \BB R$, we can deduce a Lipschitz constant $L_{x,y}=|x|\cdot |y|$ for $\ev(f)$, although there is no way to define one in a uniform way.

This observation suggests to replace the Lipschitz condition with the \emph{local} Lipschitz condition. 
Recall that  a function $f:X\to Y$ between two metric spaces is \emph{locally Lipschitz continuous} when for all $x\in X$ there exists a constant $L_{x}$ such that the inequality $d(f(y),f(z))\leq L_{x}\cdot d(y,z)$ holds in some open neighborhood of $x$.

\begin{remark}
From now on we will suppose that quantales $(Q,\geq)$ are \emph{continuous} as lattices, and we indicate by 
$\alpha  \ll \beta$ the usual \emph{way below} relation. It is clear that the Lawvere quantale and all quantales obtained from it by applying products are continuous lattices.
\end{remark}

\begin{definition}
Let $(X,Q,\ff Q,a)$ and $(Y,R,\ff R,b)$ be finitary QLR. A function $f: X\to Y$ is said \emph{generalized locally Lipschitz} (in short LL), if there exists a function $\varphi:X\times Q\to R$ (called a \emph{family of LL-constants for $f$}) such that $\varphi(x,\_)$ is additive in its second variable, and the following hold for all $x\in X$:
\begin{align*}
&\forall \alpha \in \ff Q \   \varphi(x,\alpha)\in \ff R  & \text{\emph{(finiteness)}} \\
& \exists \delta_{x} \gg 0\forall y,z\in X   \ 
  a(x,y),a(x,z)\leq \delta_{x}\To \\
& \ \  a(y,z)\leq \alpha \  \To \  b(f(y),f(z))\leq \varphi(x,\alpha)  & \text{\emph{(local\ Lipschitz)}}
\end{align*}
\end{definition}

Any locally Lipschitz function $f: \BB R\to \BB R$ yields a LL-map between the finitary QLR given by the Euclidean metric.

The finitary QLR with LL maps form a category $\LLE$: the identity function $\mathrm{id}_{X}$ has the LL  constants $\lambda x\alpha.\alpha$. Moreover, the composition of  LL   functions $f:X\to Y$ and $g: Y\to Z$ is  LL: if $\varphi$ is a family of LL  constants for $f$ and $\psi$ is a family of LL constants for $g$, then the map $(x,\alpha)\mapsto \psi(f(x), \varphi(x,\alpha))$ is a family of  LL constants for $g\circ f$ (observe that identity and composition of LL constants work precisely as in $\QQ$).

One can also consider a slightly different category $\HLLip$ defined as follows. 
First, for a QLR $(X,Q,a)$, let $\simeq_{a}$ be the equivalence relation over $X$ defined by $x\simeq_{a}x'$ if $a(x,x')=0$. We indicate by $X/a$ the quotient of $X$ by $\simeq_{a}$. By definition, the QLR $(X/a,Q,a)$ is separated.

The objects of $\HLLip$ are the same as those of $\LLE$, while the arrows between $(X,Q,\ff Q,a)$ and $(Y,R,\ff R,b)$ are pairs
 $(f,\varphi)$, where $f:X\to Y$ is LL and stable under $\simeq_{a}$-classes (i.e.~$a(x,y)=0$ implies $b(f(x),f(y))=0$), and
 $\varphi$ is a family of LL-constants for $f$ and is also stable under $\simeq_{a}$-classes (i.e.~$a(x,y)=0$ implies $\varphi(x,\alpha)=\varphi(y,\alpha)$).

There is a forgetful functor $U: \HLLip \to \LLE$ given by
$U{(X,Q,\ff Q,a)}= (X/a, Q,\ff Q, Ua)$, where $Ua([x],[y])=a(x,y)$, and $U{(f,\varphi)}=\widetilde f$, where $\widetilde f([x]_{a})=[f(x)]_{b}$. 

Given finitary QLR $(X,Q,\ff Q,a)$ and $(Y,R,\ff R,b)$ we can define the two finitary QLR \longv{\\ }
$( \LLE(X,Y), R^{X}, (\ff R)^{X},b^{X})$ and $(\HLLip(X,Y), R^{X},(\ff R)^{X}, b^{X}\circ \pi_{1})$.
Observe that if $X$ and $Y$ satisfy transitivity, so do $\LLE(X,Y)$ and $\HLLip(X,Y)$, and if $Y$ is a standard metric space, $\LLE(X,Y)$ is a standard metric space, while  $\HLLip(X,Y)$ is a pseudo-metric space.

Moreover, if the QLR $X,Y,Z$ satisfy transitivity, we  can define an isomorphism \\ $\begin{tikzcd}\HLLip(Z\times X, Y)\ar[bend left=5]{r}[above]{\lambda} & \HLLip(Z, \HLLip(X,Y)) \ar[bend left=5]{l}[below]{\ev} \end{tikzcd}$ as follows:
\begin{itemize}
\item the map $\lambda(f, \varphi)=( \langle \lambda(f), \lambda_{0}(\varphi)\rangle, \lambda_{1}(\varphi))$ is defined by 
\begin{align*}
\lambda(f)(z)(x) & =f(\langle z,x\rangle) \\
\lambda_{0}(\varphi)(z)(\langle x,\alpha\rangle) & =\varphi(\langle z,x\rangle, \langle 0,\alpha\rangle)\\ 
\lambda_{1}(\varphi)(\langle z,\zeta\rangle)(x)& = \varphi(\langle z,x\rangle, \langle \zeta, 0\rangle)
\end{align*}

\item the map $\ev(\langle g,\psi\rangle, \chi)=\langle \ev(g), \ev(\psi, \chi)\rangle $ is defined by
\begin{align*}
\ev(f)(\langle z,x\rangle) &= f(z)(x) \\
\ev(\psi, \chi) (\langle z,x\rangle, \langle \zeta, \alpha\rangle) &=
\chi(\langle z,\zeta\rangle)(x)+ \psi(z)(\langle x,\alpha\rangle)
\end{align*}
\end{itemize}
%
Reflexivity and transitivity are essential for the isomorphism above to hold: for all $(f,\varphi)\in \HLLip(Z\times X, Y)$, to show that 
\begin{align*}
b(f( z,x), f( z,x))\leq \lambda_{1}(\varphi)(\langle z, 0\rangle)(x)=0
\end{align*}
 one makes essential use of the fact that 
$b(f( z,x), f( z,x))=0$ holds in $Y$. Conversely, given
$(\langle g,\psi\rangle, \chi)\in \HLLip(Z,\HLLip(X,Y))$, 
to show that
\begin{align*} 
b(f(z,x), f( z',x'))&\leq \ev(\psi,\chi)(\langle x,z\rangle, \langle c(z,z'),a(x,x')\rangle) \\
&=  \chi(\langle z, c(z,z')\rangle)(x)+\psi (z)(\langle x,a(x,x')\rangle)
\end{align*}
one makes essential use of the transitivity of $Y$ to deduce the above from $b(f(z,x), f(z',x))\leq \chi(z, c(z,z'))(x)$ and 
$b(f(z',x), f(z',x'))\leq \psi(z)(x,a(x,x'))$.

All this leads to the following result:
%

\begin{proposition}\label{prop:uu}
The full sub-category $\LLE_{\mathsf{Met}}\hookrightarrow\LLE$ of standard metric spaces is cartesian closed. The full sub-category $\LLE_{\mathsf{pMet}}\hookrightarrow \HLLip$ of pseudo-metric spaces is cartesian closed.
Moreover, the restriction of $U$ as a functor from $\LLE_{\mathsf{Met}}$ to $\LLE_{\mathsf{pMet}}$ is a cartesian closed functor.
\end{proposition}
\longv{\begin{proof}
We first check the cartesian closure of $\LLE_{\mathsf{pMet}}$. 
\begin{itemize}

\item[($\To$)] the map $\lambda(f, \varphi)=( \langle \lambda(f), \lambda_{0}(\varphi)\rangle, \lambda_{1}(\varphi))$ is defined by 
\begin{align*}
\lambda(f)(z)(x) & =f(\langle z,x\rangle) \\
\lambda_{0}(\varphi)(z)(\langle x,\alpha\rangle) & =\varphi(\langle z,x\rangle, \langle 0,\alpha\rangle)\\ 
\lambda_{1}(\varphi)(\langle z,\zeta\rangle)(x)& = \varphi(\langle z,x\rangle, \langle \zeta, 0\rangle)
\end{align*}
For all $z\in Z$, then map $\lambda_{0}(\varphi)(z)(\_,\_)$ is additive in its second variable; moreover, 
for all $z\in Z$ and $x\in X$ there is $\langle \zeta_{z},\alpha_{x}\rangle \gg 0$ (which implies $\zeta_{z}\gg 0$ and $\alpha_{x}\gg 0$) such that, whenever $c(z,z'),c(z,z'')\leq \zeta_{z}$ and $a(x,x'),a(x,x'')\leq \alpha_{x}$, $\lambda_{0}(\varphi)(z)(\langle x,a(x',x'')\rangle)\geq b(\lambda (f)(z)(x'), \lambda (f)(z)(x''))= b(f(\langle z,x'\rangle), f(\langle z,x''\rangle))$. This proves that $\langle\lambda (f), \lambda_{0}(\varphi)\rangle(z)\in \LLE_{\mathsf{Met}}(X,Y)$.

Finally, any $z$ is contained in an open ball such that, whenever $z',z''$ belong to it,  \\ 
$\lambda_{1}(\varphi)(\langle z, c(z',z'')\rangle)(x)\geq b(\lambda (f)(z')(x), \lambda (f)(z'')(x))= b(f(\langle z',x\rangle), f(\langle z'',x\rangle))$, so we can conclude that 
$\lambda (f, \varphi)\in \LLE_{\mathsf{pMet}}(Z, \LLE_{\mathsf{pMet}}(X,Y))$.

\item[($\Leftarrow$)] the map $\ev(\langle g,\psi\rangle, \chi)=\langle \ev(g), \ev(\psi, \chi)\rangle $ is defined by
\begin{align*}
\ev(f)(\langle z,x\rangle) &= f(z)(x) \\
\ev(\psi, \chi) (\langle z,x\rangle, \langle \zeta, \alpha\rangle) &=
\chi(\langle z,\zeta\rangle)(x)+ \psi(z)(\langle x,\alpha\rangle)
\end{align*}
The map $\ev(\psi, \chi)$ is additive in its second variable. In fact we have 
\begin{align*}
\ev(\psi,\chi)(\langle z,x\rangle, \langle 0, 0\rangle) & =
\chi(\langle z,0\rangle)(x)+ \psi(z)(\langle x,0\rangle)\\ & =0+0=0
\end{align*} and 
\begin{align*}
& \ev(\psi,\chi)(\langle z,x\rangle, \langle \zeta+\zeta', \alpha+\alpha'\rangle)  \\
& =\chi(\langle z,\zeta+\zeta'\rangle)(x)+ \psi(z)(\langle x,\alpha+\alpha'\rangle) \\
& =\chi(\langle z,\zeta\rangle)(x)+\chi(\langle z,\zeta')(x) + \psi(z)(\langle x,\alpha\rangle)+\psi(z)(\langle x,\alpha'\rangle) \\
& =
\chi(\langle z,\zeta\rangle)(x)+ \psi(z)(\langle x,\alpha\rangle)+\chi(\langle z,\zeta')(x) +\psi(z)(\langle x,\alpha'\rangle)
\\ & =
\ev(\psi,\chi)(\langle z,x\rangle, \langle \zeta,\alpha\rangle)+
\ev(\psi,\chi)(\langle z,x\rangle, \langle \zeta',\alpha'\rangle)
\end{align*}

Moreover, for all $z\in Z$ and $x\in X$ there exists $\zeta_{z}\gg 0, \alpha_{x}\gg0$ (which implies $\langle \zeta_{z},\alpha_{x}\rangle \gg 0$) such that whenever $c(z,z'),c(z,z'')\leq \zeta_{z}$ and $a(x,x'),a(x,x'')\leq \alpha_{x}$
\begin{align*}
&\ev(\psi,\chi)(\langle z,x\rangle, \langle \zeta, \alpha \rangle)\\
& \geq
b(f(z')(x'), f(z')(x''))+ b(f(z')(x''),f(z'')(x'')) \\
& \geq 
b(f(z')(x'), f(z'')(x'')\\
& = b(\ev(f)(\langle z',x'\rangle), \ev (f)(\langle z'',x''\rangle))
\end{align*}

%
We can thus conclude that $\ev(\langle g,\psi\rangle, \chi)\in \LLE_{\mathsf{pMet}}(Z\times X, Y)$.

 \end{itemize}
It remains to show that $\lambda$ and $\ev$ inverse each-other:
\begin{itemize}
\item on one side we have 
$$\ev( \langle \lambda (f), \lambda_{0}(\varphi)\rangle, \lambda_{1}(\varphi)\rangle) =
\langle \ev(\lambda (f)), \ev(\lambda_{0}(\varphi), \lambda_{1}(\varphi))\rangle =\langle f, \varphi\rangle
$$
since
$\ev(\lambda_{0}(\varphi), \lambda_{1}(\varphi))(\langle z,x\rangle, \langle \zeta, \alpha\rangle)=
\varphi(\langle z,x\rangle, \langle \zeta, 0\rangle)+
\varphi(\langle z,x\rangle, \langle 0, \alpha\rangle)
=
\varphi(\langle z,x\rangle, \langle \zeta, \alpha\rangle)$ by the additivity of $\varphi$.

\item on the other side we have
$$\lambda (\ev(\langle g,\psi\rangle, \chi))=
\lambda (  \ev(g), \ev(\psi, \chi))= ( \langle \lambda(\ev(g)), \lambda_{0}(\ev(\psi, \chi))\rangle, \lambda_{1}(\ev(\psi, \chi))\rangle= (\langle g,\psi\rangle, \chi)$$
since 
$\lambda_{0}(\ev(\psi, \chi))(z)(\langle x,\alpha\rangle)=
\psi(z)(\langle x, \alpha\rangle) + \chi(\langle z,0\rangle)(x)=\psi(z)(\langle x, \alpha\rangle)$ and \\
$\lambda_{1}(\ev(\psi, \chi))(\langle z, \zeta\rangle)(x)=
\psi(z)(\langle x, 0\rangle) + \chi(\langle z,\zeta\rangle)(x)=\chi(\langle z,\zeta\rangle)(x)$.
\end{itemize}

The cartesian closure of $\LLE_{\mathsf{Met}}$ is proved as follows: if $f\in \LLE_{\mathsf{Met}}(Z\times X,Y)$, then $f$ admits a family of LL-constants $\varphi$. Then for all $z\in Z$, $\lambda_{0}(\varphi)(z)$ is a family of LL-constants
for $\lambda(f)(z)$, which implies that $\lambda (f)(z)\in \LLE_{\mathsf{Met}}(X,Y)$; moreover, $\lambda_{1}(\varphi)$ is a family of LL-constants for the application $z\mapsto \lambda(f)(z)$, so we can conclude that $\lambda(f)\in \LLE_{\mathsf{Met}}(Z,Y^{X})$. 

If now $f\in \LLE_{\mathsf{Met}}(Z, Y^{X})$, then for all $z\in Z$, the set of families of LL-constants for $f(z)$ is non-empty;  by the axiom of choice, there exists then a function $\psi$ yielding, for all $z\in Z$, a family of LL-constants for $f(z)$. Moreover $f$ itself admits a family of LL-constants $\chi$.
Then the map $\ev(\psi,\chi)$ is a family of LL-constants for $\ev(f)$, so we deduce $\ev(f)\in \LLE_{\mathsf{Met}}(X,Y)$. 

It remains to prove that $U$ is a cartesian closed functor. This descends from the following facts:
\begin{itemize}
\item ${X\times Y}/a\times b\simeq( X/a)\times( Y/b)$: in fact $\langle x,y\rangle \simeq_{a\times b}\langle x',y'\rangle $ iff $x\simeq_{a}x'$ and $y\simeq_{b}y'$.

\item ${\LLE_{\mathsf{pMet}}(X,Y)}/{b^{X}} \simeq ({Y}/{b})^{({X}/{a})}$: 
first, observe that $(f,\varphi)\simeq_{b^{X}}(g,\psi)$ iff for all $x\in X$, $f(x)\simeq_{b} g(x)$ iff
for all $x,y\in X$, $a(x,y)=0$ implies $f(x)\simeq_{b}g(y)$ (since $f,g$ are stable under $\simeq_{a}$-classes).
Now, for all $\simeq_{a}$-stable functions $f,g$, let $f\sim g$ iff 
for all $x,y\in X$, $a(x,y)=0$ implies $f(x)\simeq_{b}g(y)$. Then the claim follows from the observation that the equivalence classes of $\sim$ are in bijection with the functions from $\simeq_{a}$-classes to $\simeq_{b}$-classes.
\end{itemize}
Finally, since for all pseudo-metric space $(X,Q,a)$ we have that $Ua([x],[y])=a(x,y)$, 
from $b(f(y),f(z))\leq \varphi(x, a(y,z))$ we deduce 
$Ub(Uf([y]), Uf([z]))\leq \tilde\varphi([x], Ua([y],[z]))$. We conclude then that $\tilde\varphi$ is a family of LL-constants for $Uf$.
\end{proof}
}

The category $\HLLip$ is in some sense more constructive than $\LLE$  since to show that cartesian closure of the latter one needs the
axiom of choice (see Appendix). 

%
\begin{example}
In $\LLE$ the space of locally Lipschitz functions $\LLE(\BB R, \BB R)$ is endowed with the \emph{pointwise} metric 
$d_{\mathsf{Point}}(f,g): \BB R\to \BB R_{\geq 0}$, where 
$d_{\mathsf{Point}}(f,g)(x)=d_{\mathsf{Euc}}(f(x),g(x))$.

\end{example}

\subsection{Locally Lipschitz Models}

For any cartesian closed category, $\BB C$, we let a \emph{LL-model of $\BB C$} be a cartesian closed functor 
$F: \BB C\to  \LLE_{\mathsf{pMet}}$, 
%
Observe that  a LL-model $F: \BB C\to  \LLE_{\mathsf{pMet}}$ induces a cartesian closed functor
$U\circ F:\BB C\to  \LLE_{\mathsf{Met}}$.

Concretely, a LL-model consists in the following data:
\begin{itemize}
\item for any object $X$ of $\BB C$, a finitary pseudo-metric space $(\model X, \nudel X,  \fini X , a_{X})$;

\item for any morphism $f\in \BB C(X,Y)$, a LL-map $\model f: \model X\to \model Y$ stable on the $a_{X}$-classes, and a family of LL-constants $\nudel f: \model X\times \nudel X\to \nudel Y$ for $\model f$,

\end{itemize}
where the application $f\mapsto \nudel f$, which plays the role of the derivative in this setting, satisfies a bunch of properties that we discuss in some more detail below.


We now define a concrete model of the simply typed $\lambda$-calculus over a set of locally Lipschitz functions. 
For all $n>0$, let us fix a set $\C L_{n}$ of {locally Lipschitz} functions $f: \BB R^{n}\to \BB R$ (in the usual sense), and for each $f\in \C L_{n}$, let us fix a function $\Lip(f): \BB R^{n}\to [0,+\infty)$ associating each $\vec x\in \BB R^{n}$ with a local Lipschitz constant $\Lip(f)(\vec x)$ so that when $\vec y, \vec z$ are in some open neighborhood of $\vec x$, 
$$
 |f(\vec y)- f(\vec z)| \leq \Lip(f)(\vec x)\cdot  d^{n}_{\mathsf{Euc}}(\vec y,\vec z)
$$
where $d^{n}_{\mathsf{Euc}}(\vec y, \vec z)=\sqrt{\sum_{i}{(y_{i}-z_{i})^{2}}}$. 

For any simple type $\sigma$, a finitary pseudo-metric space $(\model\sigma, \nudel \sigma, \fini{\sigma}, a_{\sigma})$ is defined by first letting
$ \model\Real = \BB R$, $\fini{\Real} =\BB R^{\infty}_{\geq 0}$, 
 $ \nudel\Real =[0,\infty]_{+} $, $ a_{\Real} =d_{\mathsf{Euc}}$ and then lifting the definition to all other types exploiting the cartesian closed structure of $\HLLip$.
 For any simple type $\sigma$, $U(\model \sigma, \nudel \sigma, \fini \sigma, a_{\sigma})$ is then a standard metric space (observe in particular that one has $Ua_{\sigma\to \tau}(f,g)(x)=Ua_{\sigma}(f(x),g(x))$).
Moreover, given a context $\Gamma=\{x_{1}:\sigma_{1},\dots, x_{n}:\sigma_{n}\}$ and a term $t$ of type $\Gamma \vdash t:\sigma$ (that we take as representative of a class of terms of type $(\prod_{i=1}^{n}\sigma_{i})\to \sigma$), the functions $\model t: \prod_{i=1}^{n}\model{\sigma_{i}} \to \model \sigma$ and $\nudel t: \prod_{i=1}^{n}\model{\sigma_{i}}\times \prod_{i=1}^{n}\nudel{\sigma_{i}}\to \nudel \sigma$ are defined by a straightforward induction on $t$. We illustrate below only the definition of $\nudel t$:
\begin{align*}
 \nudel{\TT r} (\vec x, \vec \alpha)& = 0  \\ 
  \nudel{\TT f}(\vec x, \vec \alpha) & = \Lip(f)(\vec x)\cdot(\sum \vec \alpha) \\
\nudel{x_{i}}(\vec x, \vec \alpha) & = \alpha_{i}\\
\nudel{\langle t,u\rangle}(\vec x, \vec \alpha)& =\langle \nudel t(\vec x, \vec \alpha),\nudel u(\vec x, \vec \alpha)\rangle\\
\nudel{t\pi_{i}}(\vec x, \vec \alpha) & = \pi_{i}(\nudel t(\vec x,\vec \alpha)) \\
\nudel{\lambda y.t}(\vec x, \vec \alpha) & =  \lambda y.{  \nudel t}(\vec x*y, \vec \alpha*0 ) \\
\nudel{tu}(\vec x, \vec \alpha) & =
{\nudel t}(\vec x, \vec \alpha)(\model u(\vec x))  \\
& \quad +  \model{t}_{1}(\vec x, \vec \alpha)(\model u(\vec x), \nudel u(\vec x, \vec \alpha))
\end{align*}
where recall that for $t$ of type $\tau\to \sigma$, $\model t$ is a pair $ \langle \model t_{0}, \model t_{1}\rangle$ with 
$\model t_{0}(\vec x,\vec \alpha)\in \model \sigma^{\model \tau}$ and 
$\model t_{1}(\vec x, \vec \alpha) \in \nudel \sigma^{\model \tau \times \nudel \tau}$.

\begin{theorem}[Soundness]\label{thm:llstlc}
For all simply typed term $t$ such that $\Gamma \vdash t:\tau$, $(\model t, \nudel t)\in \LLE_{\mathsf{pMet}}(\model \Gamma, \model \sigma)$. Moreover, if $t\longrightarrow_{\beta} u$, then $\model t=\model u$ and $\nudel t=\nudel u$.
%
%
%
%

\end{theorem}

Observe that since the QLR $(\model \sigma, \nudel \sigma, a_{\sigma})$ are metric spaces, the Fundamental Lemma reduces in this case to the remark that $a_{\sigma}(\model t, \model t)=0$ holds for all term $t$ of type $\sigma$.
Instead, one can prove a ``local'' version of the contextuality lemma:
%
%
%
%
%

\begin{corollary}[local contextuality of distances]
For all terms $\vdash t,u:\sigma$ there exists $\delta_{t}\in \nudel{\sigma}$, with $\delta_{t}\gg 0$, such that for all contexts $\TT C[\ ]: \sigma \vdash \tau$
$$
a_{\tau}(\model{\TT C[t]}, \model{\TT C[u]}) \leq \nudel{\TT C}( \model t, a_{\sigma}(\model t, \model u))
$$
holds whenever $a_{\sigma}(\model t,\model u)\leq \delta_{t}$.
\end{corollary}
%
%

%
%
%
%
%
%
%
%
%

%
%
%

\subsection{Lipschitz Derivatives and Cartesian Differential Categories}

Due to their different function spaces, the derivatives constructed in $\LLE_{\mathsf{pMet}}$ (i.e.~the maps $\nudel t$) behave differently with respect to the derivatives from $\QQ$. In particular, the former behave 
more closely to the derivatives found in \emph{Differential $\lambda$-Categories} \cite{BUCCIARELLI2010213} (in short D$\lambda$C), the categorical models of the differential $\lambda$-calculus \cite{ER}.

We recall that a D$\lambda$C is a  left-additive \cite{Blute2009} category $\BB C$ in which every morphism $f\in \BB C(X, Y)$ is associated with a morphism $\DDer(f)\in \BB C(X\times X,Y)$ satisfying a few axioms: the axioms (D1)-(D7) of Cartesian Differential Categories \cite{Blute2009}, plus 
an additional axiom ($\DDer$-curry) \cite{BUCCIARELLI2010213} relating derivatives and the function space.


%

We list below the properties of the application $f\mapsto \nudel f$ in a QLR model inside $\LLE_{\mathsf{pMet}}$. We let
$\lambda_{\BB C}, \ev_{\BB C}$ indicate the isomorphism $\BB C(Z\times X,Y)\simeq \BB C(Z,\BB C(X,Y))$, $\ev^{*}_{\BB C}= \ev_{\BB C}(\mathrm{id}_{\BB C(X,Y})$, and similarly 
$\ev^{*}=\ev(\mathrm{id}_{\LLE_{\mathsf{pMet}}(X,Y)})$:
%
%
\begin{itemize}
\item[(1)] $\nudel{\mathrm{id}}=\pi_{1}$, $\nudel{g \circ f}=\nudel g\circ \langle f\circ \pi_{1},\nudel f\rangle$;

\item[(2)] $\nudel f(x,0)=0$, $\nudel f(x,\alpha+\beta)=\nudel f(x,\alpha)+\nudel f(x,\beta)$;

\item[(3)] $\nudel{\pi_{1}}=\pi_{1}\circ \pi_{1}$, $\nudel{\pi_{2}}=\pi_{2}\circ \pi_{1}$;
\item[(4)] $\nudel{\langle f,g\rangle}= \langle \nudel f, \nudel g\rangle$;

\item[(5)] $\nudel{\lambda_{\BB C}(f)}= \lambda_{X} ( \nudel f \circ \langle \pi_{1}\times \mathrm{id}_{X}, \pi_{2}\times 0   \rangle)$\\ (where for $g:Z\times X\to Y$, $\lambda_{X}(g)=\lambda x.g(\langle\_,x\rangle)$)

\item[(6)] $\nudel{\ev^{*}_{\BB C}\circ\langle h,g\rangle}=
\ev^{*}\circ \langle\nudel h , g\circ \pi_{1}\rangle +
\nudel{\ev_{\BB C}(h)} \circ \langle \langle \pi_{1},g\circ\pi_{1}, \langle 0, \nudel g\rangle\rangle$,
(where $h\in\BB C(Z,\BB C(X,Y))$, $g\in \BB C(Z,X)$).
\end{itemize}


The properties above literally translate the fact that a QLR model is a cartesian closed functor:
\begin{itemize}
\item (1) says that $f\mapsto \nudel f$ is functorial;
\item (2) says that $\nudel f$ is additive in its second variable;
\item (3) and (4) say that the cartesian structure of $\BB C$ commutes with that of $\LLE_{\mathsf{pMet}}$;
\item (5) and (6) say that the cartesian closed structure of $\BB C$ commutes with that of $\LLE_{\mathsf{pMet}}$.

\end{itemize} 

 (1)-(2)-(3)-(4) coincide with axioms (D2)-(D3)-(D4)-(D5) of Cartesian Differential Categories (in short, CDC). Actually, this is not very surprising, since these axioms describe the fact that the application $f\mapsto \langle f, \DDer(f)\rangle$ in a CDC $\BB C$ yields a cartesian functor (known as the \emph{tangent functor}, see \cite{Cockett2011}).
Observe that the other axioms of CDCs do not make sense in our setting, because $\LLE_{\mathsf{pMet}}$ is not left-additive and   there are no ``second derivatives'' in $\LLE_{\mathsf{pMet}}$.

Finally, property (5) is precisely axiom ($\DDer$-curry) of D$\lambda$Cs, and property (6) can be deduced in any D$\lambda$C from the other axioms (cf.~\cite{BUCCIARELLI2010213}, Lemma 4.5). 
%
%
%
%
%
%
%
%

\section{Related Works}

Logical relations \cite{Plotkin1973, STATMAN198585} are a standard method to establish program equivalence and other behavioral properties of higher-order programs, also related to the concept of \emph{relational parametricity} \cite{Reynolds1983}.
The primary source of inspiration for the QLR 
are differential logical relations (DLR) \cite{dallago,dallago2}, whose cartesian closed structure is very similar to that of the category $\QQ$. 
While DLR can be seen as special cases of QLR (see footnote \ref{foot1}), 
%
the only metric structure studied for the DLR in \cite{dallago} are 
what we called here hyper-relaxed metrics.
A precursor of this approach is \cite{chaudhuri}, which develops a System F-based system for approximate program transformations, but without explicitly mentioning any metric structure. 

The category $\VPM$ from Section \ref{section5} is reminiscent of the diameter spaces from \cite{Geoffroy2020}, which form a cartesian \emph{lax}-closed category based on a similar factorization of partial metric spaces. A main difference is that in \cite{Geoffroy2020} the factorization  is considered as a \emph{property} of (suitable) partial metric spaces, rather than an additional \emph{structure}, as we do here.

Several \emph{relational logics} have been developed to formalize logical relations and, more generally, higher-order relational reasoning \cite{Plotkin1993, Ahmed2011, 10.1145/2775051.2676980, 10.1145/3009837.3009877,10.1145/3110265}, including quantitative reasoning \cite{Barthe_2012, Gabo2019b}.
An important question, which transcends the scope of this paper, is whether one can describe a QLR semantics for at least some of these logics, or if a different relational logic has to be developed in order to capture  
quantitative relational reasoning based on QLR.

The literature on program metrics in denotational semantics is vast. 
Since \cite{ARNOLD1980181} metric spaces have been exploited as an alternative framework to standard, domain-theoretic, denotational semantics. Notably,  \emph{Banach's fixed point theorem} plays the role of standard order-theoretic fixpoint theorems in this setting (see \cite{VANBREUGEL20011} and \cite{BAIER1994171}). 

More recently, program metrics have been applied in the field of differential privacy \cite{10.1145/1932681.1863568, 10.1007/978-3-642-29420-4_3, Barthe_2012}, by relying on Lipschitz-continuity as 
a foundation for the notion of program sensitivity. To this line of research belongs also the literature on System $\mathsf{Fuzz}$ \cite{10.1145/1932681.1863568}, a sub-exponential PCF-style language designed for differential privacy, which admits an elegant semantics based on metric spaces and metric CPOs \cite{10.1145/1932681.1863568, Gaboardi2017}.

Ultra-metrics are widely applied in program metrics, mostly to describe intensional aspects (e.g.~traces, computation steps) \cite{VANBREUGEL20011, MAJSTERCEDERBAUM1991217, Escardo1999}, also for the $\lambda$-calculus, due to the fact that when $Q$ is a locale, $\mathsf{Met}_{Q}$ is cartesian closed. 
%

Partial metrics were introduced in \cite{matthews} with the goal of modeling partial objects in program semantics, and independently discovered in sheaf theory as \emph{$M$-valued sets} \cite{Hohle:1992aa}. 
\cite{Bukatin1997} shows that partial metrics and relaxed metrics can be used to characterize the topology of continuous Scott domains with a countable bases. This work was, to our knowledge, the first to acknowledge the correspondence between partial metrics and lattices, which was later developed through the theory of valuations \cite{10.1007/BFb0053546, ONeill, 10.1016/j.tcs.2003.11.016}. \cite{AGT7849} provides a topological characterization of partial metric spaces.
Fuzzy and probabilistic partial metric spaces are well-investigated too \cite{Yueli:2015aa, Wu:2017aa, HE201999}.
 Our description of generalized partial metric spaces was based on the elegant presentation from \cite{Stubbe2018, STUBBE201495} in the language of quantaloid-enriched categories.

Together with standard real-valued metrics, 
Lawvere's generalized metrics \cite{Lawvere1973} have also played a major role in these research lines. 
More generally, the abstract investigation of metric spaces as quantale and quantaloid-enriched categories is part of the growing field of \emph{monoidal topology} \cite{Hofmann2014}. To this approach we can ascribe the already mentioned description of partial metric spaces from \cite{Stubbe2018, STUBBE201495}, as well as the general characterization of \emph{exponentiable} metric spaces and quantaloid-enriched categories in \cite{CLEMENTINO20063113,Clementino:2009aa}. 

Quantitative approaches based on generalized metric spaces have been developed for bisimulation metrics \cite{Bonchi2014, Bonchi2018, DBLP:journals/lmcs/BaldanBKK18} and algebraic effects \cite{Plotk, 10.1145/3209108.3209149}. 
Generalized metrics based on Heyting quantales have been used to investigate properties of graphs and transition systems (see \cite{Pouzet2020} for a recent survey).

%


Finally, research on axiomatizations of abstract notions of differentiation has been a very active domain of research in recent years \cite{Blute2009, Cockett2011, Cockett2018, Blute2019,10.1007/978-3-030-17127-8_3, 10.1007/978-3-030-45231-5_4}, supported by the growth of interest in algorithms based on automatic differentiation. 
The two notions of derivative discussed in this paper can be compared with two lines of research on abstract differentiation.
On the one hand, the derivatives arising from differential logical relations (which essentially coincide with the derivatives from $\QQ$) have been compared \cite{dallago2} with those found in some recent literature on discrete differentiation (e.g.~finite difference operators, Boolean derivatives), and approaches based on the so-called \emph{incremental $\lambda$-calculus}
 \cite{Cai2014, 10.1007/978-3-030-17184-1_19, 10.1007/978-3-030-45231-5_4}.
 On the other hand, the derivatives from Section \ref{section6} can be compared with the literature on Cartesian Differential Categories, originating  in Ehrhard and Regnier's work on \textit{differential linear logic} and the \textit{differential $\lambda$-calculus} \cite{ER}.
Very recently, Cartesian Difference Categories \cite{10.1007/978-3-030-45231-5_4} have been proposed as a framework unifying these two lines of research.

\section{Conclusion}

%
%

This paper provides just a first exploration of the program metrics semantics that arise from the study of  quantitative logical relations, and leaves a considerable number of open questions. 
We indicate a few natural prosecutions of this work.

While our focus here was only on cartesian closure, it is natural to look for QLR-models with further structure (e.g.~coproducts, recursion, monads etc.). For instance, by extending the picture to \emph{quantaloid}-valued relations \cite{STUBBE201495}, one can define a coproduct of QLR with nice properties.

The correspondence between metrics and enriched categories suggests to consider the transitivity axiom as a ``vertical'' composition law for distances. 
An interesting question is whether one can define higher-dimensional categories of program distances with a nice compositional structure, in analogy with well-investigated higher-dimensional models in  \emph{categorical rewriting} \cite{Meseguer1992, Miyoshi1996}.
At a more formal level, the same observation also suggests to investigate \emph{relational logics} to formalize the metric reasoning justified by QLR-models, in line with the program logics developed for standard logical relations \cite{Plotkin1993, Ahmed2011} and for quantitative relational reasoning \cite{10.1145/2775051.2676980, 10.1145/3009837.3009877,10.1145/3110265, Barthe_2012, Gabo2019b}.


%
%
%

\bibliographystyle{plain}
\bibliography{main.bib}

\end{document}